\newif\ifarxiv
\arxivtrue

\documentclass[conference]{IEEEtran}
\usepackage{fancyhdr}
\IEEEoverridecommandlockouts
\usepackage{cite}
\usepackage{amsmath,amssymb,amsfonts}
\usepackage{graphicx}
\usepackage{textcomp}
\usepackage{xcolor}
\def\BibTeX{{\rm B\kern-.05em{\sc i\kern-.025em b}\kern-.08em
    T\kern-.1667em\lower.7ex\hbox{E}\kern-.125emX}}

\usepackage{microtype}
\usepackage{hyperref}
\usepackage{url}
\usepackage{listings}
\usepackage{subfig}
\usepackage{multirow}
\usepackage{booktabs}
\usepackage{graphicx}
\usepackage{xspace}
\usepackage{algorithm}
\usepackage[noend]{algpseudocode}
\usepackage{balance}
\usepackage{physics}
\usepackage{siunitx}
\usepackage{wrapfig}
\usepackage{amsthm}
\usepackage{threeparttable}
\usepackage{cleveref}
\usepackage{comment}

\algblock{ParFor}{EndParFor}
\algnewcommand\algorithmicparfor{\textbf{parfor}}
\algnewcommand\algorithmicpardo{\textbf{do}}
\algrenewtext{ParFor}[1]{\algorithmicparfor\ #1\ \algorithmicpardo}
\algtext*{EndParFor}

\newtheorem{theorem}{Theorem}

\newtheorem{problem}{Problem}

\newtheorem{constraint}{Constraint}
\newtheorem{definition}{Definition}

\newcommand{\cuq}{cuQuantum\xspace}
\newcommand{\hyq}{HyQuas\xspace}
\newcommand{\qis}{Qiskit\xspace}
\newcommand{\snq}{SnuQS\xspace}
\newcommand{\qdao}{QDAO\xspace}
\newcommand{\uniq}{UniQ\xspace}
\newcommand{\Sys}{Atlas\xspace}
\newcommand{\sys}{\Sys}
\newcommand{\Syss}{\Sys'\xspace} 
\newcommand{\kr}[2]{{#1} \downarrow {#2}}
\newcommand{\partit}{$\Call{Partition}{}$\xspace}
\newcommand{\stage}{$\Call{Stage}{}$\xspace}
\newcommand{\execute}{$\Call{Execute}{}$\xspace}
\newcommand{\simulate}{$\Call{Simulate}{}$\xspace}
\newcommand{\dpordered}{$\Call{OrderedKernelize}{}$\xspace}
\newcommand{\dporderedplain}{\textsc{OrderedKernelize}\xspace}
\newcommand{\kernelize}{$\Call{Kernelize}{}$\xspace}
\newcommand{\kernelizeplain}{\textsc{Kernelize}\xspace}
\newcommand{\extq}{ExtQ}

\newcommand{\kwins}{insular\xspace}
\newcommand{\kwnonins}{non-insular\xspace}
\newcommand{\Kwins}{Insular\xspace}

\newcommand{\m}{\mathcal}

\newcommand{\ZJ}[1]{\textcolor{blue}{ZJ: #1}}
\newcommand{\MX}[1]{\textcolor{red}{MX: #1}}
\newcommand{\UA}[1]{\textcolor{red}{UA: #1}}

\newcommand{\defn}[1]{\emph{#1}}
\newcommand{\KS}[1]{\ensuremath{\kappa}}
\newcommand{\tcd}{\texttt}

\begin{document}

\ifarxiv
\title{
\Sys: Hierarchical Partitioning for Quantum Circuit Simulation on GPUs (Extended Version)
}
\fi

\makeatletter
\newcommand{\linebreakand}{%
  \end{@IEEEauthorhalign}
  \hfill\mbox{}\par
  \mbox{}\hfill\begin{@IEEEauthorhalign}
}
\makeatother

\author{
\IEEEauthorblockN{Mingkuan Xu}
\IEEEauthorblockA{\textit{Computer Science Department} \\
\textit{Carnegie Mellon University}\\
Pittsburgh, PA, USA \\
mingkuan@cmu.edu}
\and
\IEEEauthorblockN{Shiyi Cao}
\IEEEauthorblockA{\textit{Department of EECS} \\
\textit{UC Berkeley}\\
Berkeley, CA, USA \\
shicao@berkeley.edu}
\and
\IEEEauthorblockN{Xupeng Miao}
\IEEEauthorblockA{\textit{Computer Science Department} \\
\textit{Carnegie Mellon University}\\
Pittsburgh, PA, USA \\
xupeng@cmu.edu}
\linebreakand
\IEEEauthorblockN{Umut A. Acar}
\IEEEauthorblockA{\textit{Computer Science Department} \\
\textit{Carnegie Mellon University}\\
Pittsburgh, PA, USA \\
umut@cmu.edu}
\and
\IEEEauthorblockN{Zhihao Jia}
\IEEEauthorblockA{\textit{Computer Science Department} \\
\textit{Carnegie Mellon University}\\
Pittsburgh, PA, USA \\
zhihao@cmu.edu}
}



\maketitle
\ifarxiv
\thispagestyle{fancy}
\lhead{}
\rhead{}
\chead{}
\lfoot{\footnotesize{This is the extended version of a paper presented in SC24~\cite{xu2024atlas}.
\newline This version includes an additional appendix with detailed results.}}
\rfoot{}
\cfoot{}
\renewcommand{\headrulewidth}{0pt}
\renewcommand{\footrulewidth}{0pt}
\fi
\begin{abstract}
  %
  %
  %
  %

  This paper presents techniques for theoretically and practically
  efficient and scalable Schr\"odinger-style quantum circuit simulation. 
  Our approach partitions a quantum circuit into a hierarchy of
  subcircuits and simulates the subcircuits on multi-node GPUs, exploiting available data parallelism while minimizing
  communication costs.
  To minimize communication costs, we formulate an Integer Linear
  Program that rewards simulation of ``nearby'' gates on ``nearby'' GPUs.
  To maximize throughput, we use a dynamic programming
  algorithm to compute the subcircuit simulated by each kernel at a GPU.
  %
  %
  %
  We realize these techniques in \sys, a distributed, multi-GPU quantum circuit simulator.
  %
  %
  Our evaluation on a variety of quantum circuits shows that \Sys outperforms state-of-the-art GPU-based simulators by more than 2$\times$ on average and is able to run larger circuits via offloading to DRAM, outperforming other large-circuit simulators by two orders of magnitude.

\end{abstract}


\section{Introduction}
\label{sec:intro}


Quantum computing has established an advantage over classical
computing, especially in areas such as cryptography, machine learning,
and physical
sciences~\cite{nielsen2000quantum,bennett2020quantum,orus2019quantum,biamonte2017quantum,lloyd2013quantum,schuld2019quantum,aspuru2005simulated,babbush2015chemical}.
Several quantum computers have been built recently, including Sycamore~\cite{arute2019quantum}, Bristlecone~\cite{bristlecone}, Jiuzhang~\cite{zhong2020quantum},
Osprey~\cite{Choi_2023},
and Condor~\cite{Castelvecchi_2023}.
However, the robustness demands of many quantum applications exceed those
of modern {\em noisy intermediate-scale quantum} (NISQ) computers, which
suffer from decoherence and lack of error-correction~\cite{Preskill_2018, Corcoles_2020}.
Furthermore, NISQ computers are expensive resources---many remain inaccessible beyond a small group.
Therefore, there is significant interest in quantum circuit simulation,
which enables performing robust quantum computation on classical parallel
machine.

Several approaches to simulating quantum circuits have been proposed,
including Schr\"odinger- and Feynman-style
simulation.
Even though Feynman-style simulation can require a small amount of
space, its time requirement appears very high. Therefore, most modern
simulators use Schr\"odinger-style quantum circuit simulation, which maintains an entire quantum state in a \defn{state vector} of size $2^n$ for $n$ qubits, and
applies each gate of the circuit to it.

%
%

%
To tackle the scalability challenges of state-vector simulation, researchers exploit the plethora of
parallelism available in this task by using CPUs, GPUs, and parallel
machines~\cite{amariutei2011parallel, avila2014gpu,
avila2016optimizing, qiskit2019, gutierrez2010quantum, jones2019quest,
zhang2015quantum,smelyanskiy2016qhipster,suzuki2021qulacs,park_snuqs_2022,chen2021hyquas}.
To cope with the increasing demand for simulating larger 
circuits, recent work has proposed techniques for storing state
vectors on DRAM and secondary storage (e.g.,
disks)~\cite{haner_05_2017,smelyanskiy2016qhipster}.
Even after over a decade of research, quantum circuit simulation
continues to remain challenging for performance and scalability.

At a high level of abstraction, there are three key challenges to
(Schr\"odinger style) state-vector simulation.
The first challenge is the space requirements of storing the state
vector with $2^n$ amplitudes, each of which is a complex
number.
Overcoming this challenge requires distributing the state vector in a parallel machine
with heterogeneous memory.
{Second, simulating the application of a quantum gate to one or multiple qubits involves strided accesses to the state vector, which requires {\em fine-grained} inter-node and/or inter-device communications. The cost of these communications becomes the performance bottleneck for large-scale circuit simulation.}
%
{Third, because quantum gates typically operate on a small number of qubits, simulating a quantum gate involves multiplying a sparse matrix with the state vector,  resulting in low arithmetic intensity, e.g., as little as 0.5 on real-world quantum circuits~\cite{haner_05_2017}. This harms performance on modern accelerators (e.g., GPUs) designed for compute-intensive workload.}

In this paper, we propose techniques for 
performant and scalable quantum circuit simulation on modern GPU
systems.
Specifically, we consider a multi-node GPU architecture, where each
node may host a number of GPUs.
We then perform on this architecture quantum circuit simulation in the
(Schr\"odinger) state-vector style by using the available memory across
all nodes to store the state.
To minimize the communication costs, we partition the quantum circuit
into a number of \defn{stages}, each of which consists of a
(contiguous) subcircuit of the input circuit that can be simulated on
a single GPU without requiring non-local accesses outside the memory
of the GPU.
To maximize the benefits from parallelism, we further partition the
subcircuit of each stage into \defn{kernels} or groups of gates that
are large enough to benefit from parallelism but also small enough to
prevent exponential blowups in cost.

Given such a partitioned circuit consisting of stages and kernels, we
present a simulation algorithm that performs the simulation in stages
and that restricts much of the expensive communication to take place
only between the stages.
To maximize throughput and parallelism within each stage, the
algorithm \defn{shards} the state vector into
contiguous pieces such that each shard may be used to simulate the
subcircuit of the stage on a single GPU.
The approach allows multiple shards (of a stage) to be executed in
parallel or sequentially depending on the availability of resources.

We formulate the partitioning problem consisting of staging and kernelization and
present provably effective algorithms for solving them.
For the staging problem, we provide a solution based on Integer Linear
Programming, which can be solved by an off-the-shelf solver.
For the kernelization problem, we present a dynamic programming
solution that can ensure optimality under some assumptions.

We present an implementation, called \Sys, that realizes the proposed
approach.
%
%
Our evaluation on a variety of quantum circuits on up to 256 GPUs (on
64 nodes) shows that \Sys is up to $20.2\times$ faster than
state-of-the-art GPU-based simulators and can scale to large quantum
circuits that go beyond GPU memory capacity (up to $105\times$ faster than state-of-the-art simulators that also go beyond GPU memory capacity). 

In summary, this paper makes the following contributions.

\begin{itemize}
\item A hierarchical partitioning approach to scaling
  performant quantum circuit simulation based on
  staging and kernelization.

\item An ILP algorithm to stage a circuit for simulation that can
  minimize the number of stages.

\item A dynamic programming algorithm for kernelizing each stage to
  ensure efficient parallelism.
      
\item An implementation that realizes hierarchical partitioning and significantly outperforms existing simulators.
\end{itemize}

\section{Background}
\label{sec:background}


\textbf{Architectural Model.}
We assume a multi-node GPU architecture with $2^G$ nodes.
Each node contains multiple GPUs and a
single CPU with an attached DRAM module.
Each GPU can store in its \defn{local} memory $2^L$ amplitudes
(complex numbers).
Each node can store in its \defn{regional} memory $2^{L+R}$
amplitudes, where $2^R$ can be the number of GPUs per node or can be
set such that $2^{L+R}$ equals the number of amplitudes that can be
stored in the DRAM memory of the node.

\textbf{Quantum Information Fundamentals.}
A quantum circuit consists of a number of qubits (quantum bits) and gates that operate on them.
%
%
The state $\ket{\psi}$ of a $n$-qubit  
quantum circuit is
a superposition of its \defn{basis states} denoted $\ket{0}$,
$\ket{1}$, ..., $\ket{2^n-1}$ and is typically written as
$
\ket{\psi} = \sum_{i=0}^{2^n-1} \alpha_i \ket{i},
$
where $\alpha_i$ is a complex coefficient (also called amplitude) of
the basis state $\ket{i}$.
When measuring the state of the system, the probability of observing the
state $\ket{i}$ as the output is $|\alpha_i|^2$; therefore,
$\sum_{i=0}^{2^n-1} |\alpha_i|^2 = 1$.
When simulating a quantum circuit with~$n$ qubits, we can represent
its state with a vector of $2^n$ complex values $\vec{\alpha} =
(\alpha_0,\alpha_1,...,\alpha_{2^n-1})^\top$.


The semantics of a $k$-qubit gate is specified by a $2^k \times 2^k$
unitary complex matrix $U$ and applying the gate to a quantum circuit
with state $\ket{\psi}$ results in a new state: $\ket{\psi}
\rightarrow U \ket{\psi}$.  For example, applying a 1-qubit gate on
the $q$-th qubit of a quantum system updates its state vector as
follows:

\begin{equation}
\begin{bmatrix}
\alpha_{f(i)}\\
\alpha_{f(i)+2^q}
\end{bmatrix}
\rightarrow
\begin{bmatrix}
u_{00} & u_{01} \\
u_{10} & u_{11}
\end{bmatrix}
\times
\begin{bmatrix}
\alpha_{f(i)}\\
\alpha_{f(i)+2^q}
\end{bmatrix},
\label{eqn2}
\end{equation}
where $\begin{bmatrix}
u_{00} & u_{01} \\
u_{10} & u_{11}
\end{bmatrix}$
is the $2\times 2$ unitary complex matrix that represents the 1-qubit gate, and $f(i) =  2^{q+1} \lfloor\frac{i}{2^q}\rfloor + (i \bmod 2^q)$ for all integers between 0 and $2^{n-1}-1$.

\begin{figure*}[ht!]
    \centering
    \includegraphics[scale=0.48]{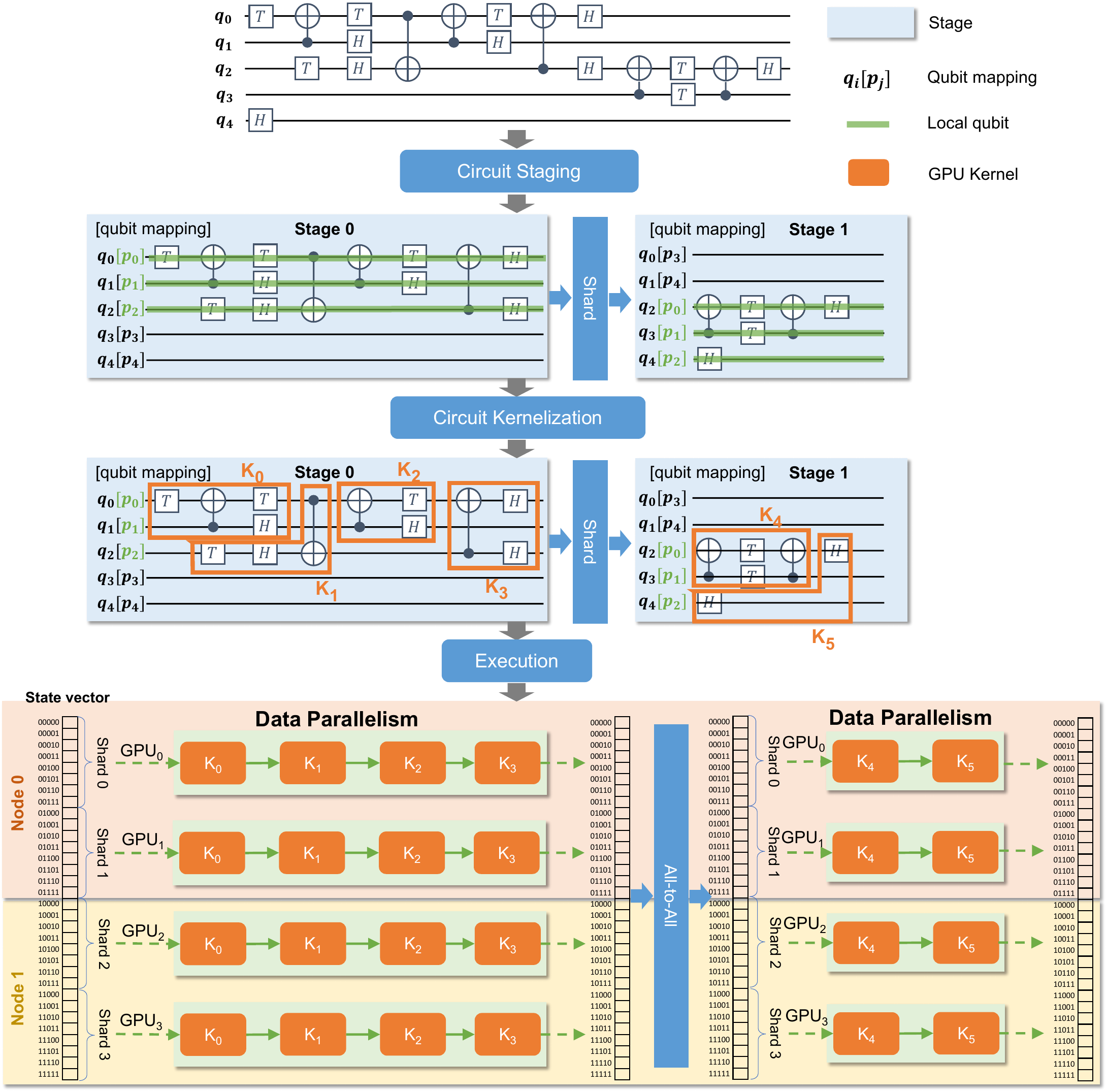}
    \caption{An example application of circuit partitioning and execution. We stage the circuit so that qubits of each gate map to local qubits (i.e., green lines in each stage).  The notation $q_i[p_j]$ indicates that the $i$-th logical qubit maps to the $j$-th physical qubit. The \kernelizeplain algorithm then partitions the gates of each stage into kernels that provide for data parallelism.
    }
    \label{fig:atlas}
\end{figure*}

\section{Hierarchical Partitioning for Simulation}
\label{sec:atlas}

We present our quantum simulation algorithm with hierarchical circuit
partitioning.


\begin{algorithm}
\begin{algorithmic}[1]
        
\Function{Partition}{$\m{C}, L, R, G$}
\State $stages = \Call{Stage}{\m{C}, L, R, G}$
\State $stagedKernels$ = []
\For {$i = 0 \dots |stages|-1$}
\State $(\m{C}_i, \m{Q}) = stages[i]$
\State $kernels_i = \Call{Kernelize}{\m{C}_i}$
\State $stagedKernels[i] = (kernels_i, \m{Q})$
\EndFor
\State \Return $stagedKernels$
\EndFunction

\Function{Execute}{$stagedKernels, state$}    
\State $\m{P}  = $ identity permutation
\State $shards = [state]$
\For {$i = 0 \dots |stagedKernels|-1$}
\State $(kernels_i, \m{Q}) = stagedKernels[i]$
\State $(shards, \m{P}) = \Call{Shard}{shards, \m{Q}, \m{P}}$
\ParFor {$shard$ \textbf{in} $shards$}
\For {$k$ \textbf{in} $kernels_i$}
\State $\Call{LaunchKernel}{k, shard}$
\EndFor
\EndParFor
\EndFor
\EndFunction

\Function{Simulate}{$\m{C}, state, L, R, G$}    
\State $stagedKernels = \Call{Partition}{\m{C}, L, R, G}$
\State $\Call{Execute}{stagedKernels, state}$
\EndFunction

\end{algorithmic}
\caption{Hierarchical partitioning of a quantum circuit $\m{C}$ and
  the simulation algorithm for simulating the circuit starting with
  the input state $state$. The parameters $L, R, G$ describe the distributed execution model with $L$ local qubits, $R$ regional qubits, and $G$ global qubits.}
\label{alg:atlas}
\end{algorithm}

\Cref{alg:atlas} shows the pseudocode for our algorithms for quantum
circuit simulation with hierarchical partitioning.
The algorithm \partit takes as arguments the input circuit
$\m{C}$ and the architecture parameters consisting of the number of
local, regional, and global qubits, $L, R, G$ respectively.
It partitions the input circuit $\m{C}$ into stages using the
\stage function.
Each \defn{stage} consists of a subcircuit and a partition of the
logical qubits of the (entire) circuit into sets of local, regional,
and global qubits, such that the subcircuit can be simulated entirely
locally by a single GPU, performing local memory accesses only.
We describe the staging algorithm \stage in \Cref{sec:staging}.

After staging the circuit, the algorithm \partit, proceeds
to kernelize each circuit by calling  \kernelize
on the subcircuit of each stage.
The algorithm \kernelize further partitions the subcircuit for
the stage into subcircuits each of which may be efficiently executed
on a single GPU by taking advantage of data parallelism.
The basic idea behind the algorithm is to group gates to amortize the
cost of data parallelism on modern GPUs (e.g., due to cost of
launching kernels) and do so without leading to an exponential blow-up of
costs, because grouping gates can increase the cost exponentially in the
number of qubits involved.
We describe the kernelization algorithm \kernelize in \Cref{sec:kernelizer}.
The algorithm \partit completes by returning a list of
stages, each of which consists of a list of kernels.

The algorithm \execute takes a list of kernelized stages and
input state consisting of a vector of amplitudes, and executes each
stage in order.
For each stage, the algorithm permutes the state vector and shards it
into contiguous sections as demanded by the qubit partition $\m{Q}$
for that stage.
To minimize communication costs during this permutation and sharding
process, the algorithm keeps track of the current permutation $\m{P}$.
The algorithm then considers each shard in parallel and applies the
kernels of the stage to the shard sequentially.
Because our algorithm staged and kernelized the circuit to maximize
throughput for a distributed GPU architecture, each kernel can run
efficiently on a single GPU, by applying each kernel to all the
amplitudes in a shard of the state vector in a data parallel fashion.
If sufficiently many GPUs are available, then each shard may be
assigned to a GPU for parallel application.
If, however, fewer GPUs are available, the shards may be stored in the
shared DRAM at each node, and swapped in and out of the GPUs for
execution.

Our simulation algorithm builds on top of our \partit and
\execute algorithms.
Given the input circuit $\m{C}$ and an input state $state$, along with
architecture parameters $L, R, G$, our simulation algorithm
\simulate starts by partitioning the input circuit into a list of
kernelized stages.
It then calls the algorithm \execute to execute the stages.

\Cref{fig:atlas} shows an example application of the algorithm. We note that \partit does not depend on $state$. Both \stage and \kernelize work for arbitrary input states.

\section{Circuit Staging}
\label{sec:staging}

We present an algorithm for (circuit) staging that partitions the
circuit into \defn{stages}, each of which is a contiguous subcircuit
of the input circuit, along with a partition of the qubits into local,
regional, and global sets such that each gate in the subcircuit of the
stage operates on the local qubits.

\begin{definition}[Local, regional, and global qubits]
\label{def:logical_regional_global}
For a quantum circuit simulation with $n$ qubits, let $q_i$ ($0\leq i \leq n-1)$ be the $i$-th physical qubit. Let each shard include $2^L$ states, and let the DRAM of each node save $2^{L+R}$ states.

\begin{itemize}

\item The first $L$ physical qubits (i.e., $q_0,...,q_{L-1}$) are {\em local} qubits.
\item The next $R$ physical qubits (i.e., $q_{L},...,q_{L+R-1}$) are {\em regional} qubits. 
\item The final $G = n - L - R$ physical qubits (i.e., $q_{L+R},...,q_{n-1}$) are {\em global} qubits.
\end{itemize}

\end{definition}

We categorize physical qubits into local, regional, and global subsets based on the communications required to simulate the application of a general quantum gate to these qubits.
First, applying a gate to a local qubit only requires accessing states within the same shard, and therefore avoids communications.
Second, applying a gate to a regional qubit requires accessing states in different shards stored on the same compute node, which only requires inter-device (intra-node) communications.
Finally, applying a gate to a global qubit requires states stored on different compute nodes and thus involves inter-node commutations.

Given this partition of local, regional, and global qubits, we aim to map logical qubits of a circuit to physical qubits 
in a way that avoids excessive communication.
To minimize communication cost, we leverage a specific type of qubits in certain types of gates, called {\em insular qubits}.
%

\begin{definition}[\Kwins{} Qubit]
\label{def:insular_qubit}
For a single-qubit gate, the qubit is \kwins if the
unitary matrix of the gate is diagonal or anti-diagonal, i.e., the
non-zero entries are along the (anti) diagonal.
For multi-qubit controlled-$U$ gates\footnote{A {\em controlled}-$U$ gate has some {\em control} qubits controlling a $U$ gate (can be any unitary gate) on some {\em target} qubits. 
If at least one of the control qubits is $\ket{0}$, then the target
remains unchanged.
If all of the control qubits are $\ket{1}$, then the gate $U$ is applied to
the target qubits.}, all control qubits are \kwins\footnote{In some controlled-$U$ gates, any qubit can be chosen as the control qubit without changing the output. In such gates, all qubits are \kwins{}.}.
%
All other qubits are \kwnonins{}.
\end{definition}

\begin{figure}[t]
    \centering
    \subfloat[Sharding with inter-node\\ communication.]{
    \includegraphics[scale=0.33]{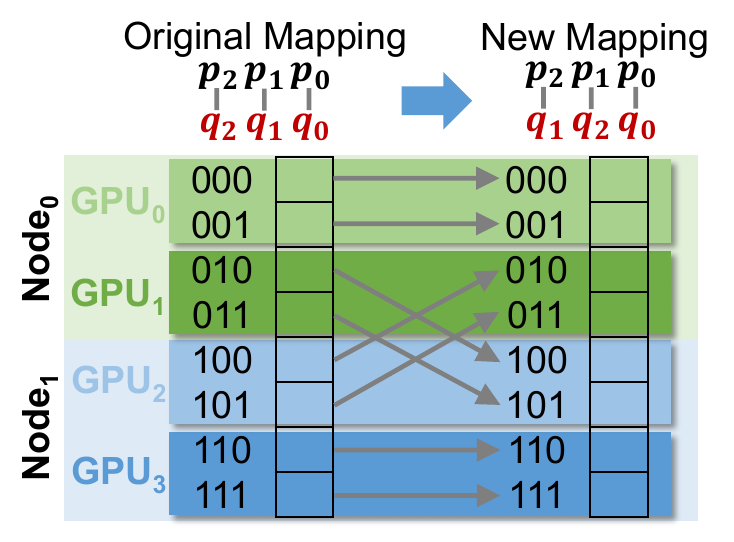}
    \label{fig:inter_node_communications}
    }
    \subfloat[Sharding with intra-node\\(inter-GPU) communication.]{
    \includegraphics[scale=0.33]{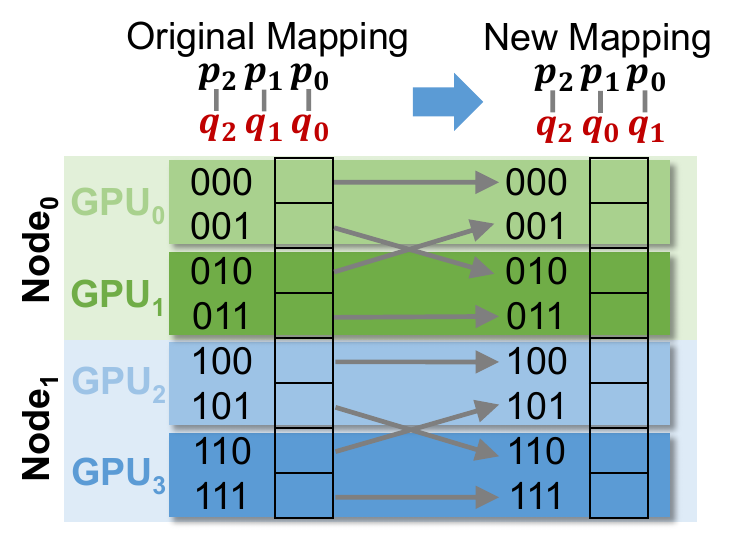}
    \label{fig:inter_device_communications}
    }
    \caption{Sharding with different types of communication. The simulation has 1 local, 1 regional, and 1 global qubit. $p_i - q_j$ indicates that the $i$-th physical qubit is mapped to the $j$-th logical qubit. Inter-node communication is triggered if we update any global qubits (\Cref{fig:inter_node_communications}), and only intra-node communication is triggered otherwise (\Cref{fig:inter_device_communications}).}
    \label{fig:qubit_remapping}
\end{figure}

Our idea of \kwins{} qubits was inspired by {\em global gate specialization} introduced by H\"aner \emph{et~al.}~\cite{haner_05_2017}.
Intuitively, for a single-qubit gate with an \kwins{} qubit, computing each output state only depends on one input state (since the gate's unitary matrix is diagonal or anti-diagonal), which allows \Sys to map \kwins{} qubits to regional and/or global physical qubits without introducing any extra communication. This property allows \Sys to only consider the \kwnonins{} qubits of quantum gates when mapping qubits.

\Sys staging algorithm (\Cref{alg:ilp}) splits a quantum circuit's simulation into multiple stages, each of which includes a subcircuit and uses a different mapping from logical to physical qubits.
Within each subcircuit, all \kwnonins{} qubits of all gates can only operate on \emph{local} physical qubits. 
This approach avoids any inter-device or inter-node communications within a stage and only requires all-to-all communications between stages to perform qubit remapping.



We formalize circuit staging as a constrained optimization problem and design a binary integer linear programming (ILP) algorithm to discover a staging strategy that minimizes the total communication cost.



\paragraph{Circuit staging problem}
Given an inter-node communication cost factor $c$, for a given input circuit $\m{C}$ with $n$ qubits and an integer $s$, \Sys circuit staging algorithm splits $\m{C}$ into at most $s$ stages $\m{C}_0,..., \m{C}_{s-1}$ and determines a qubit partition $\m{Q}$ of local/regional/global qubits for each stage, minimizing the total communication cost:
\begin{equation}
\sum_{i=1}^{s-1} \left(\left|\m{Q}_i^{local} \setminus \m{Q}_{i-1}^{local}\right| + c \cdot \left|\m{Q}_i^{global} \setminus \m{Q}_{i-1}^{global}\right|\right).
\label{eqn:ilp_objective_informal}
\end{equation}
where $\m{Q}_i^{local}$ is the local qubit set of stage $i$ and $\m{Q}_i^{global}$ is the global qubit set of stage $i$. $\left|\m{Q}_i^{local} \setminus \m{Q}_{i-1}^{local}\right|$ is the number of local qubits that need to be updated, approximating the inter-GPU communication cost; $\left|\m{Q}_i^{global} \setminus \m{Q}_{i-1}^{global}\right|$ is the number of global qubits that need to be updated, approximating the extra inter-node communication cost.
Although regional qubits do not appear directly in \eqref{eqn:ilp_objective_informal}, they are critical for allowing the number of global qubits that need to be updated to be smaller than the local one. Without regional qubits, any inter-GPU communication will also be inter-node.


Let $\m{G}$ denote the gates of the circuit $\m{C}$, and let $\m{E}$ denote their dependencies (adjacent gate pairs on the same qubit).
Let $A_{q,k}$ and $B_{q,k}$ ($0\leq q < n$, $0 \leq k < s$) denote whether the $q$-th logical qubit is mapped to a local/global physical qubit at the $k$-th stage, and let $F_{g,k}$ ($g \in \m{G}$, $0 \leq k < s$) denote if the gate $g$ is finished by the end of the $k$-th stage. 
For each $A_{q,k}$ and $B_{q,k}$ ($k < s-1$), we also introduce two ancillary variables $S_{q,k}$ and $T_{q,k}$ that indicate whether the $q$-th logical qubit is updated from local to non-local or from global to non-global from the $k$-th to the $(k+1)$-th stage (i.e., $S_{q,k} =1$ if and only if $A_{q,k} = 0$ and $A_{q,k+1} = 1$; $T_{q,k} =1$ if and only if $B_{q,k} = 0$ and $B_{q,k+1} = 1$).
Note that $A_{q,k}$, $B_{q,k}$, $F_{g,k}$, $S_{q,k}$, and $T_{q,k}$ are all binary variables.
Minimizing the total cost of local-to-non-local and global-to-non-global updates with constraints yields the following objective:

\begin{equation}
    \min \sum_{k=0}^{s-2} \sum_{q=0}^{n-1} \left(S_{q,k} + c \cdot T_{q,k}\right)\label{eqn:objective}
\end{equation}
subject to
\begin{align}
&A_{q,k+1}\leq A_{q,k} + S_{q,k} & \forall q \in [0,n)\ \ \forall k \in [0,s-1) \label{c1}\\
&B_{q,k+1}\leq B_{q,k} + T_{q,k} & \forall q \in [0,n)\ \ \forall k \in [0,s-1) \label{cdeft} \\
&F_{g,k}\leq F_{g,k+1} & \forall g \in \m{G}\ \ \forall k \in [0,s-1) \label{c2}\\
&F_{g, k}\leq F_{g,k-1} + A_{q,k} & \textrm{$q$ is a non-insular qubit of $g$}\label{c3}\\ 
&F_{g_1,k}\geq F_{g_2,k} & \forall (g_1, g_2) \in \m{E}\ \ \forall k \in [0,s)\label{c4}\\
&F_{g,s-1}=1 & \forall g \in \m{G} \label{c5} \\
&A_{q,k} + B_{q,k} \le 1 & \forall q \in [0,n)\ \ \forall k \in [0,s) \label{cag} \\
&\sum_{q=0}^{n-1}{A_{q,k}}=L ~~ \sum_{q=0}^{n-1}{B_{q,k}}=G \hspace{-8ex}& \forall k \in [0,s) \label{c6} 
\end{align}


We now describe these constraints. First, constraints~\ref{c1}, \ref{cdeft} and \ref{c2} can be derived from the definitions of $S_{q,k}$, $T_{q,k}$ and $F_{g,k}$. Second, constraint~\ref{c3} is a locality constraint that a gate $g$ can be executed at the $k$-th stage only if $g$'s non-insular qubits $q$ are all mapped to local physical qubits (i.e., $A_{q,k} = 1$).
Third, constraint~\ref{c4} represents a dependency constraint that all gates must be executed by following a topological order with respect to dependencies between them.
Fourth, constraint~\ref{c5} is derived from the completion constraint that all gates must be executed in the $s$ stages.
Next, constraint~\ref{cag} declares that any qubit cannot be both local and global at the same time.
Finally, constraint~\ref{c6} specifies the hardware constraint that each stage has $L$ local physical qubits and $G$ global qubits.

\paragraph{ILP-based circuit staging.} For a given input circuit $\m{C}$ and an integer $s$ that specifies the maximum number of stages $\m{C}$ can be partitioned into, \Sys sends the objective (i.e., \Cref{eqn:objective}) and all constraints (i.e., \Cref{c1,cdeft,c2,c3,c4,c5,cag,c6}) to an off-the-shelf integer-linear-programming (ILP) solver, which returns an assignment for matrices $A,B,F,S,T$ that minimizes the objective while satisfying all constraints. 
\Sys then stages the circuit $\m{C}$ into $(\m{C}_0,\m{Q}_0),\dots,(\m{C}_{s-1},\m{Q}_{s-1})$ based on $A$, $B$ and $F$, where $\m{C}_i$ is a subcircuit and $\m{Q}_i$ is a qubit partition into local/regional/global qubit sets.
Specifically, for a quantum gate $g \in \m{G}$, $g$ will be executed by the stage with index: $\min\{k | F_{g,k} = 1\}$. 
\Sys maps the $q$-th logical qubit to a local (physical) qubit at the $k$-th stage if $A_{q,k} = 1$, to a global qubit at the $k$-th stage if $B_{q,k} = 1$, and to a regional qubit if $A_{q,k} = B_{q,k} = 0$.



\Cref{alg:ilp} shows \Sys staging algorithm. It uses the ILP solver as a subroutine and returns the first feasible result.

\begin{algorithm}[tb]
\small
\begin{algorithmic}[1]
\Function{Stage}{$\m{C},L,R,G$}
\For{$s = 1, 2, \dots$} \label{ilp:loop_s}
\State $status, A, B, F = \Call{SolveILP}{\m{C},L,R,G,s}$
\If{$status = feasible$}
\State Find $(\m{C}_0,\m{Q}_0),\dots,(\m{C}_{s-1},\m{Q}_{s-1})$ based on $A$, $B$, $F$
\State \Return $(\m{C}_0,\m{Q}_0),\dots,(\m{C}_{s-1},\m{Q}_{s-1})$
\EndIf
\EndFor
\EndFunction
\end{algorithmic}
\caption{\Sys circuit staging algorithm.}
\label{alg:ilp}
\end{algorithm}

\begin{theorem}[Optimality of \stage]
\Cref{alg:ilp} returns the minimum feasible number of stages. \label{thm:ilp-optimal}
\end{theorem}
\begin{proof}
\Cref{alg:ilp} returns the first feasible result of the ILP solver, and the loop in line~\ref{ilp:loop_s} ensures that the number of stages is minimized.
The ILP constraints are the feasible constraints for stages, so \Cref{alg:ilp} returns the minimum feasible number of stages.
\end{proof}
On top of the minimum number of stages, the ILP minimizes the total communication cost (the ILP objective).

\section{Circuit Kernelization}
\label{sec:kernelizer}


After partitioning an input circuit into multiple stages, \Sys must
efficiently execute the gates of each stage on GPUs.
GPU computations are organized as {\em kernels}, each of which is
a function simultaneously executed by multiple hardware threads in a
single-program-multiple-data (SPMD) fashion~\cite{darema1988spmd}.
A straightforward approach to executing gates is to launch a kernel
for each gate, which yields suboptimal performance due to the low arithmetic intensity of each gate. Another extreme is to fuse all gates into a single kernel, which would exponentially increase the amount of computation.

To improve simulation performance on GPUs, \Sys uses a dynamic programming algorithm (\Cref{alg:dp}) to partition the gates of a stage into kernels.
We formulate circuit kernelization as an optimization problem as follows.
%


\begin{problem}[Optimal circuit sequence kernelization]

Given a cost function $\Call{Cost}{\cdot}$ that maps one or multiple kernels (a sequence of gates) to their cost (i.e., a real number),
the optimal circuit sequence kernelization problem takes as input
a quantum circuit $\m{C}$ represented as a sequence of gates,
and
returns a kernel sequence $\m{K}_0,...,\m{K}_{b-1}$ that minimizes
the total execution cost:
\begin{equation}
\sum_{i=0}^{b-1} \Call{Cost}{\m{K}_i}
\label{eqn:dp_objective}
\end{equation}
such that each kernel consists of a contiguous segment of the gate sequence $\m{C}$.
\label{problem:kernelization}
\end{problem}




We note that summing up the cost of individual kernels
(\Cref{eqn:dp_objective}) is a faithful representation of the total
cost, because on a GPU, each kernel may execute in parallel, but
multiple kernels are executed sequentially.
%

Also, note that each kernel consisting of a contiguous segment of the gate sequence $\m{C}$ is a conservative requirement.
It is feasible to execute any kernel sequence that forms a topologically equivalent sequence to $\m{C}$.
However, there are exponentially many ways to consider kernels with non-contiguous segments.
To consider them efficiently, we present a key constraint on the kernels of the \kernelize algorithm in \Cref{alg:dp}.
\ifarxiv
We include a study of an algorithm that only considers contiguous kernels in \Cref{app:dp-simple}.
\fi

\begin{figure}
\centering
\includegraphics[scale=0.55]{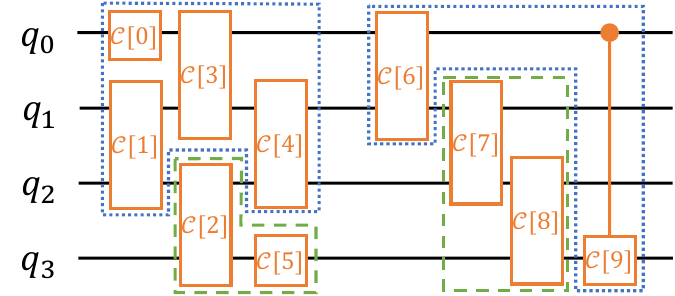}
\caption{Kernel examples. The two green dashed kernels satisfy \Cref{dp-constraint}. The two blue dotted kernels do not satisfy \Cref{dp-constraint} and thus are not considered by the \kernelizeplain algorithm.}
\label{fig:kernel-constraint}
\end{figure}

\begin{constraint}[Constraints on kernels]
Each kernel $\m{K}$ considered by the \kernelize algorithm satisfies the following constraints:
\begin{enumerate}
\item Weak convexity: $\forall\,j_1 < j_2 < j_3$, if $\m{C}[j_1]
  \in \m{K}$, $\m{C}[j_2] \notin \m{K}, \m{C}[j_3] \in \m{K}$, then
  $\Call{Qubits}{\m{C}[j_1]}\cap\Call{Qubits}{\m{C}[j_2]}$ $\cap\,
  \Call{Qubits}{\m{C}[j_3]} = \emptyset$.  Informally, weak convexity
  requires that for any three gates $\m{C}[j_1], \m{C}[j_2],
  \m{C}[j_3]$, if only the middle gate is not in the kernel, then they
  cannot share a qubit.

\item Monotonicity: for any gate $\m{C}[j] \not\in \m{K}$, if
  $\m{C}[j]$ shares a qubit with $\kr{\m{K}}{j}$ then
  $\Call{Qubits}{\m{K}} = \Call{Qubits}{\kr{\m{K}}{j}}$ where
  $\kr{\m{K}}{j} = \m{K} \cap \{\m{C}[0],\m{C}[1],\dots,\m{C}[j-1]\}$.
  Informally, if we decide to exclude a gate from $\m{K}$ while it shares a qubit
  with $\m{K}$, then we fix the qubit set of $\m{K}$ from that gate
  on. 
\end{enumerate}
\label{dp-constraint}
\end{constraint}
In \Cref{fig:kernel-constraint}, the blue dotted kernel on the left violates weak convexity because $\m{C}[1]$, $\m{C}[2]$ and $\m{C}[4]$ share $q_2$, and only $\m{C}[2]$ is not in the kernel. If we allow this kernel to be considered, it will be mutually dependent with the kernel containing $\m{C}[2]$, yielding no feasible results.
The blue dotted kernel on the right does not violate weak convexity, but it violates monotonicity. $\m{C}[7]$ is excluded from the kernel while sharing the qubit $q_1$ with the kernel, so the qubit set of the kernel should be fixed to $\{q_0, q_1\}$. So it cannot include $\m{C}[9]$. In fact, including $\m{C}[9]$ in this kernel causes mutual dependency with the kernel $\{\m{C}[7],\m{C}[8]\}$.

\Cref{dp-constraint} ensures the kernels to form a valid sequence in \Cref{thm:dp-correct}.

\begin{algorithm}[t]
{
\small
\begin{algorithmic}[1]
\Function{$\Call{Kernelize}{}$}{$\m{C}$}

\State {\bf Input:} A quantum circuit $\m{C}$ represented as a sequence of gates.
\State {\bf Output:} A sequence of kernels.
\State // Suppose that we have an ordered kernel set $\KS{}$ in the first $i$ gates, $DP[i,\KS{}]$ stores the minimum cost to kernelize the first $i$ gates \emph{except for the gates in $\KS{}$} and the corresponding kernel sequence.
\State $DP[i,\KS{}] = (\infty, [])$ for all $i, \KS{}$
\State $DP[0,\emptyset] = (0, [])$
\For {$i = 0$ \textbf{to} $|\m{C}|-1$} \label{dp:loop_i}
\For {\textbf{each} $\KS{}$ \textbf{such that} $\exists\, \m{K},\,\m{C}[i] \in \m{K}\in \KS{}$ \textbf{and} $\m{K}$ satisfies \Cref{dp-constraint}} \label{dp:loop_ks}
\State {$\tilde{\KS{}} = \KS{} \setminus \{\m{K}\}$}
\If {$|\m{K}|>1$}
\State {$DP[i+1,\KS{}] = DP[i,\tilde{\KS{}} \cup \{\m{K} \setminus \{\m{C}[i]\}\}]$} \label{dp:add-directly}
\Else \ // $\m{K} = \{\m{C}[i]\}$
\State {$DP[i+1,\KS{}] = \displaystyle\min_{\KS{}' \supseteq \tilde{\KS{}}}\left\{\!DP[i,\KS{}'] + \Call{Cost}{\KS{}'\setminus\tilde{\KS{}}}\right\}$} \label{dp:singleton-kernel}
\EndIf
\EndFor
\EndFor
\State $ DP_{best} = \min_{\KS{}}\left\{DP[|\m{C}|,\KS{}]+\Call{Cost}{\KS{}}\right\}$ \label{dp:compute_f_best}
\State \Return $DP_{best}.kernels$
\EndFunction
\end{algorithmic}
}
\caption{The \kernelizeplain algorithm.
The operators ``+'' on lines~\ref{dp:singleton-kernel} and~\ref{dp:compute_f_best} mean adding the cost and appending the kernel set to the sequence. 
We maintain the order in place in line~\ref{dp:add-directly} (replacing $\m{K} \setminus \{\m{C}[i]\}$ with $\m{K}$) if the monotonicity constraint in \Cref{dp-constraint} applies to $\m{K}$, or move $\m{K}$ to the end of the kernel set otherwise.
We require $\tilde{\KS{}}$ to be a suffix of $\KS{}'$ in line~\ref{dp:singleton-kernel}.}
\label{alg:dp}
\end{algorithm}

\begin{theorem}[Correctness of the \kernelize algorithm]
\Cref{alg:dp} returns a kernel sequence such that concatenating the kernels forms a sequence topologically equivalent to $\m{C}$.
\label{thm:dp-correct}
\end{theorem}

\begin{proof}
%
%
Proof by induction on $i$ that for each $DP[i,\KS{}].cost < \infty$, appending $\KS{}$ to $DP[i,\KS{}].kernels$ results in the sequence $\kr{\m{C}_{\KS{}}}{i}$ which is topologically equivalent to a the sequence $\m{C}[0], \m{C}[1]$, $\dots, \m{C}[i-1]$ (denote this sequence $\kr{\m{C}}{i}$). This holds for $i = 0$ where $DP[i,\emptyset].kernels$ and $\kr{\m{C}}{0}$ are both empty.

Suppose the induction hypothesis holds for $i$. For $i+1$, for any $\m{K}$ satisfying \Cref{dp-constraint}, if $|\m{K}| = 1$, line~\ref{dp:singleton-kernel} in the algorithm makes the order of appending $\KS{}$ to $DP[i+1,\KS{}].kernels$ the same as the order of appending $\KS{}'$ to $DP[i,\KS{}'].kernels$ for some $\KS{}'$ when $\tilde{\KS{}}$ is a suffix of $\KS{}'$, so we can directly apply the induction hypothesis to prove it for $i+1$.

Suppose $|\m{K}| > 1$. If the monotonicity constraint in \Cref{dp-constraint} applies to $\m{K}$, line~\ref{dp:add-directly} just inserts $\m{C}[i]$ to the end of $\m{K}$ in $\kr{\m{C}_{\KS{}}}{i}$. 
So it suffices to prove that $\m{C}[i]$ does not depend on any gate after the end of $\m{K}$ in $\kr{\m{C}_{\KS{}}}{i}$, i.e., $\m{C}[i]$ does not share any qubit with these gates.

By monotonicity, \begin{equation}
  \Call{Qubits}{\m{K}} = \Call{Qubits}{\kr{\m{K}}{i}}. \label{eqn:dp-correct-monotonicity}
  \end{equation}

For any gate $\m{C}[j_2]$ after the last gate of $\m{K}$ in the sequence $\kr{\m{C}_{\KS{}}}{i}$ before the insertion,
by weak convexity, $\forall\,j_1 < j_2 < i$ such that $\m{C}[j_1] \in \m{K}$ (we know that $\m{C}[j_2] \notin \m{K}, \m{C}[i] \in \m{K}$), $\Call{Qubits}{\m{C}[j_1]}\cap\Call{Qubits}{\m{C}[j_2]}\cap
\Call{Qubits}{\m{C}[i]} = \emptyset$. Because all gates in $\m{K}$ are before $\m{C}[j_2]$ in the sequence $\kr{\m{C}_{\KS{}}}{i}$ before the insertion,
\begin{equation}
  \Call{Qubits}{\kr{\m{K}}{i}}\cap\Call{Qubits}{\m{C}[j_2]}\cap
\Call{Qubits}{\m{C}[i]} = \emptyset \label{eqn:dp-correct-convexity}
\end{equation} where $\kr{\m{K}}{i}$ is the kernel $\m{K}$ before the insertion of $\m{C}[i]$.

By \Cref{eqn:dp-correct-monotonicity},
\begin{equation}
\Call{Qubits}{\m{C}[i]} \subseteq \Call{Qubits}{\kr{\m{K}}{i}}.
\end{equation}
Plugging this into \Cref{eqn:dp-correct-convexity},
\begin{equation}
\Call{Qubits}{\m{C}[j_2]}\cap
\Call{Qubits}{\m{C}[i]} = \emptyset
\end{equation}
for any gate $\m{C}[j_2]$ after the last gate of $\m{K}$ in the sequence $\kr{\m{C}_{\KS{}}}{i}$ before the insertion.

So we can safely insert $\m{C}[i]$ to the end of $\m{K}$ in $\kr{\m{C}_{\KS{}}}{i}$ without breaking the topological equivalence.

If the monotonicity constraint in \Cref{dp-constraint} does not apply to $\m{K}$, line~\ref{dp:add-directly} moves $\m{K}$ to the end of $\kr{\m{C}_{\KS{}}}{i}$. Because the monotonicity constraint does not apply to $\m{K}$, for any gate $\m{C}[j_0] \in \m{K}$, for any gate $\m{C}[j] \notin \m{K}$ after $\m{C}[j_0]$ in the original sequence $\kr{\m{C}_{\KS{}}}{i}$, $\m{C}[j]$ and $\m{C}[j_0]$ do not share any qubits. So we can safely move the entire kernel $\m{K}$ to the end without breaking the topological equivalence.

In all cases, we have proved that $\kr{\m{C}_{\KS{}}}{(i+1)}$ is topologically equivalent to $\kr{\m{C}}{(i+1)}$. So $\kr{\m{C}_{\KS{}}}{|\m{C}|}$ is topologically equivalent to $\kr{\m{C}}{|\m{C}|}$, which proves that any state in line~\ref{dp:compute_f_best} in the algorithm returns a kernel sequence such that concatenating the kernels forms a sequence topologically equivalent to $\m{C}$.
\end{proof}

\begin{theorem}[\Cref{dp-constraint} allows all contiguous kernels]
Any kernel of a contiguous segment of the gate sequence $\m{C}$ satisfies \Cref{dp-constraint}. \label{thm:contiguous}
\end{theorem}
\begin{proof}
For any kernel $\m{K}$ of a contiguous segment of gates, for any $j_1 < j_2 < j_3$ such that $\m{C}[j_1] \in \m{K}, \m{C}[j_2] \notin \m{K}$, we know that $\m{C}[j_3] \notin \m{K}$ by contiguity. So weak convexity holds, and monotonicity holds because $\kr{\m{K}}{j_2} = \m{K}$.
\end{proof}

\begin{figure}
    \centering
    \includegraphics[scale=0.55]{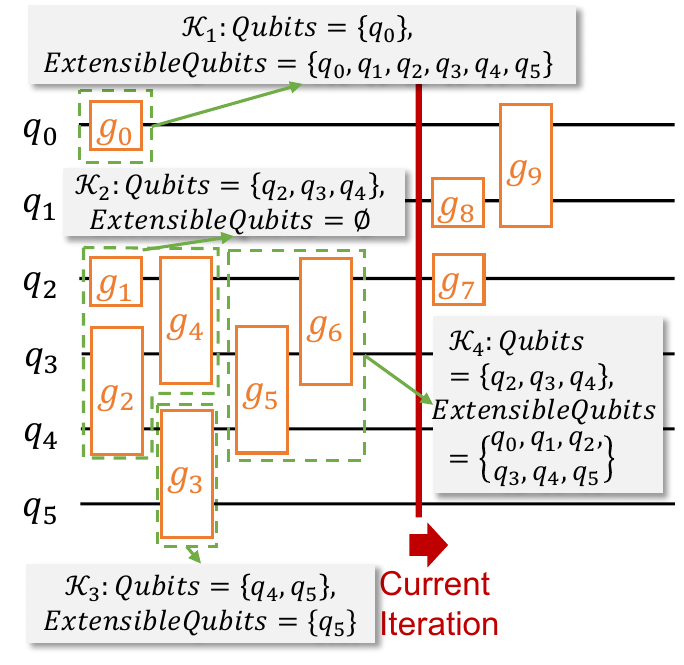}
    \caption{An example DP state in the implementation. The circuit sequence is $\m{C} = [g_0, g_1, \dots, g_9]$.}
    \label{fig:dp_with_optimization}
\end{figure}

\section{Implementation}
\subsection{Reducing Size of DP State in \kernelizeplain}
\label{subsec:dp-simplify}
In a quantum circuit, the number of gates is typically orders of magnitude greater than the number of qubits.
Therefore, instead of tracking the gates for each kernel in \Cref{alg:dp}, we maintain the qubit set as well as an \defn{extensible qubit set} for each kernel. 

\begin{algorithm}[t]
{
\small
\begin{algorithmic}[1]
\State // $\m{C}[i] \in \m{K}\in \KS{}$
\If{$\m{K} = \{\m{C}[i]\}$} // a new kernel
\State $\Call{\extq}{\m{K},i+1} = \Call{AllQubits}{}$
\Else
\State $\Call{\extq}{\m{K},i+1} = \Call{\extq}{\m{K},i}$
\EndIf
\For{$\m{K}' \in (\KS{} \setminus \m{K})$}
\If {$\Call{\extq}{\m{K}',i} = \Call{AllQubits}{}$} \label{extensible:if-all}
\If{$\Call{Qubits}{\m{C}[i]} \cap \Call{Qubits}{\m{K}'} \neq \emptyset$} \label{extensible:intersect}
\State $\Call{\extq}{\m{K}',i+1} = 
\Call{Qubits}{\m{K}'} \setminus \Call{Qubits}{\m{C}[i]}$ \label{extensible:qubits-k}
\Else
\State $\Call{\extq}{\m{K}',i+1} = \Call{AllQubits}{}$
\EndIf
\Else
\State $\Call{\extq}{\m{K}',i+1} = \Call{\extq}{\m{K}',i} \setminus \Call{Qubits}{\m{C}[i]}$ \label{extensible:remove-qubits}
\EndIf
\EndFor
\end{algorithmic}
}
\caption{Maintaining the extensible qubit set 
$\textsc{\extq}$
for each kernel after each iteration of the for loop in line~\ref{dp:loop_ks} of \Cref{alg:dp}.}
\label{alg:extensible}
\end{algorithm}

\begin{definition}[Extensible qubit for a kernel]
For a kernel $\m{K}$ at position $i$, i.e., $\kr{\m{K}}{i}$, a qubit $q$ is extensible if and only if adding a gate operating on qubit $q$ to $\kr{\m{K}}{i}$ satisfies \Cref{dp-constraint}.
Formally, a qubit $q$ is extensible for $\kr{\m{K}}{i}$ if and only if both of the following hold:
\begin{enumerate}
\item Weak convexity: $\forall\, j_1 < j_2 < i$, if $\m{C}[j_1] \in \kr{\m{K}}{i}$, $\m{C}[j_2]$ $\notin \kr{\m{K}}{i}$, then $q \notin \Call{Qubits}{\m{C}[j_1]} \cap \Call{Qubits}{\m{C}[j_2]}$.
\item Monotonicity: for any gate $\m{C}[j] \notin \kr{\m{K}}{i}$, if $j < i$ and $\m{C}[j]$ shares a qubit with $\kr{\m{K}}{j}$, then $q \in \Call{Qubits}{\kr{\m{K}}{j}}$.
\end{enumerate}
\label{def:extensible-qubit}
\end{definition}
This definition is a direct characterization of \Cref{dp-constraint}: adding the gate $\m{C}[i]$ to a kernel $\m{K}$ satisfies \Cref{dp-constraint} if and only if all qubits of $\m{C}[i]$ are extensible for $\kr{\m{K}}{i}$. For example, the extensible qubit set is $\{q_5\}$ for $\m{K}_3$ in \Cref{fig:dp_with_optimization} -- we can only add gates operating only on this qubit to $\m{K}_3$ to satisfy \Cref{dp-constraint}. 



We append the code fragment of \Cref{alg:extensible} to the end of each iteration of the for loop in line~\ref{dp:loop_ks} of \Cref{alg:dp} to maintain the extensible qubit set for each kernel.

\begin{theorem}
\Cref{alg:extensible} correctly computes the extensible qubit set for each kernel.
\label{thm:extensible-correct}
\end{theorem}

\begin{proof}
Proof by induction on the position $i$.
For the kernel $\m{K}$ such that $\m{C}[i] \in \m{K}$, if $\m{K} = \{\m{C}[i]\}$, all qubits are extensible by \Cref{def:extensible-qubit}. Otherwise, $\Call{\extq}{\m{K}, i}$ is computed in the previous iteration.

$\forall\,\m{K}' \in (\KS{} \setminus \{\{\m{C}[i]\}\})$, suppose $\Call{\extq}{\m{K}',i}$ is computed correctly.

Case 0. Consider the kernel $\m{K}$ such that $\m{C}[i] \in \m{K}$ while $|\m{K}|>1$. In the condition for weak convexity, $\m{C}[j_2]$ cannot be $\m{C}[i]$ because it requires $\m{C}[j_2] \notin \kr{\m{K}}{(i+1)}$, so this constraint is exactly the same as weak convexity for $\kr{\m{K}}{i}$.

If monotonicity applies to $\kr{\m{K}}{i}$, it still applies to $\kr{\m{K}}{(i+1)}$, and $\Call{Qubits}{\kr{\m{K}}{(i+1)}} = \Call{Qubits}{\kr{\m{K}}{i}}$; if monotonicity does not apply to $\kr{\m{K}}{i}$,  $\m{C}[j]$ cannot be $\m{C}[i]$ in the condition because it requires $\m{C}[j] \notin \kr{\m{K}}{(i+1)}$, so it still does not apply to $\kr{\m{K}}{(i+1)}$.
So $\Call{\extq}{\m{K},i+1} = \Call{\extq}{\m{K},i}$.

Now we case over each kernel $\m{K}' \in (\KS{} \setminus \m{K})$. If monotonicity applies to $\kr{\m{K}'}{i}$, there will be a gate $\m{C}[j] \notin \kr{\m{K}'}{i}$ satisfying $j < i$ and $\Call{Qubits}{\m{C}[j]} \cup \Call{Qubits}{\kr{\m{K}'}{j}} \neq \emptyset$, causing all qubits in this intersection to be not extensible for $\kr{\m{K}}{i}$ by weak convexity (picking $j_2 = j$). So if the condition in line~\ref{extensible:if-all} of \Cref{alg:extensible} holds, i.e., all qubits are extensible for $\kr{\m{K}'}{i}$, then monotonicity does not apply to $\kr{\m{K}'}{i}$.

Case 1. If the condition in line~\ref{extensible:intersect} of \Cref{alg:extensible} does not hold, i.e., $\Call{Qubits}{\m{C}[i]} \cap \Call{Qubits}{\kr{\m{K}'}{(i+1)}}=\emptyset$, then monotonicity does not apply to $\kr{\m{K}'}{(i+1)}$. Consider the additional constraint of weak convexity for $\kr{\m{K}'}{(i+1)}$ than weak convexity for $\kr{\m{K}'}{i}$, that is, $\forall\, j_1 < j_2 = i$, if $\m{C}[j_1] \in \kr{\m{K}'}{(i+1)}, \m{C}[i] \notin \kr{\m{K}'}{(i+1)}$, then $q \notin \Call{Qubits}{\m{C}[j_1]} \cap \Call{Qubits}{\m{C}[i]}$. However, $\Call{Qubits}{\m{C}[j_1]} \cap \Call{Qubits}{\m{C}[i]} = \emptyset$, so 
$\Call{\extq}{\m{K}',i+1} = \Call{AllQubits}{}$.

Case 2. If the condition in line~\ref{extensible:intersect} of \Cref{alg:extensible} holds, i.e., $\Call{Qubits}{\m{C}[i]} \cap \Call{Qubits}{\kr{\m{K}'}{(i+1)}}\neq \emptyset$,
then monotonicity applies to $\kr{\m{K}'}{(i+1)}$ because $\m{C}[i] \notin \kr{\m{K}'}{(i+1)}$ and $\m{C}[i]$ shares a qubit with $\kr{\m{K}'}{(i+1)}$. So only qubits in $\Call{Qubits}{\kr{\m{K}'}{(i+1)}}$ can possibly be extensible.

$\forall\,q \in \Call{Qubits}{\kr{\m{K}'}{(i+1)}}$, $q$ is extensible if and only if it satisfies weak convexity. Because all qubits were extensible for $\kr{\m{K}'}{i}$, we only need to consider the constraint on $\forall j_1 < j_2 = i$.
Because $q \in \Call{Qubits}{\kr{\m{K}'}{(i+1)}}$, we can always find some $j_1 < i$ satisfying $q \in \Call{Qubits}{\m{C}[j_1]}$ and $\m{C}[j_1] \in \kr{\m{K}'}{(i+1)}$, so $q$ is extensible if and only if $q \notin \Call{Qubits}{\m{C}[i]}$. Therefore, $\Call{\extq}{\m{K}',i+1} = 
\Call{Qubits}{\kr{\m{K}'}{(i+1)}} \setminus \Call{Qubits}{\m{C}[i]}$.

Case 3. If the condition in line~\ref{extensible:if-all} of \Cref{alg:extensible} does not hold, i.e., not all qubits are extensible for $\kr{\m{K}'}{i}$, we know that line~\ref{extensible:qubits-k} has been executed for some $i' < i$ (so monotonicity applies to $\kr{\m{K}'}{(i'+1)}$), so $\Call{\extq}{\m{K}',i'+1}\subseteq \Call{Qubits}{\kr{\m{K}'}{i'}}$. Because line~\ref{extensible:remove-qubits} of \Cref{alg:extensible} ensures $\Call{\extq}{\m{K}',i''+1} \subseteq \Call{\extq}{\m{K}',i''}$ for all $i' < i'' < i$, we have $\Call{\extq}{\m{K}', i}\subseteq \Call{Qubits}{\kr{\m{K}'}{i'}} \subseteq \Call{Qubits}{\kr{\m{K}'}{i}}$.

For each qubit $q$, if it is not in $\Call{\extq}{\m{K}',i}$, it cannot be in $\Call{\extq}{\m{K}',i+1}$ because \Cref{def:extensible-qubit} is more restrictive when $i$ increases; for each $q \in \Call{\extq}{\m{K}',i}$, $q \in \Call{\extq}{\m{K}',i+1}$ if and only if it satisfies the additional weak convexity constraint $\forall\, j_1 < j_2 = i$.
Because $\Call{\extq}{\m{K}',i}\subseteq \Call{Qubits}{\kr{\m{K}'}{i}}$, we can always find $j_1 < i$ such that $q \in \Call{Qubits}{\m{C}[j_1]}$ and $\m{C}[j_1]\in \kr{\m{K}}{i}$, so $q$ is extensible if and only if $q \notin \Call{Qubits}{\m{C}[i]}$.
Therefore, we conclude that $\Call{\extq}{\m{K}',i+1} = \Call{\extq}{\m{K}',i} \setminus \Call{Qubits}{\m{C}[i]}$.

In all of these cases, we have proved the correctness of $\Call{\extq}{\m{K}',i+1}$ for any $\m{K}'$. So \Cref{alg:extensible} correctly computes the extensible qubit set for each kernel.
\end{proof}

Because we cannot add any gates to kernels with an empty extensible qubit set such as $\m{K}_2$, we compute their cost and remove them from $\KS{}$ in line~\ref{dp:singleton-kernel} of \Cref{alg:dp}.
Because line~\ref{extensible:qubits-k} of \Cref{alg:extensible} restricts the extensible qubit set to be a subset of the qubit set of $\m{K}'$, we no longer need to keep track of the qubit set because it will no longer change.
%

\begin{theorem}[Time complexity of the \kernelize algorithm]
With the extensible qubit sets maintained in \Cref{alg:extensible}, the \kernelize algorithm runs in $O\left(\left(\frac{3.2 n}{\ln{n}}\right)^n \cdot |\m{C}|\right)$, where $n$ is the number of qubits.
\label{thm:dp-complexity}
\end{theorem}
\begin{proof}
For each iteration $i$ in \Cref{alg:dp}, let us compute the maximum number of kernel sets $\KS{}$ that \Sys must consider.

We observe that if we consider the qubit set for each kernel with the extensible qubit set being all qubits, and the extensible qubit set for other kernels, any two of these sets do not intersect. This can be proven by induction on $i$ and the fact that whenever we add $\Call{Qubits}{\m{C}[i]}$ to $\m{K}$, we effectively remove these qubits from all other sets.

Because kernels with empty extensible qubit sets are removed from $\KS{}$, there are at most $n$ kernels in $\KS{}$ and the number of ways to partition all qubits into disjoint qubit sets is at most the Bell number $B_{n+1}$, which is less than $\left(\frac{0.792(n+1)}{\ln(n+2)}\right)^{n+1}$~\cite{berend_improved_2010}. Each qubit set can be either the qubit set or the extensible qubit set of a kernel, so we need to multiply this number by at most $2^n$; each kernel can be either fusion or shared-memory, so we need to multiply this number by another $2^n$.
With the extensible qubit sets maintained, the order of kernels does not affect the DP transitions, so the number of different kernel sets is at most $4^n B_{n+1}$.
Each iteration of the for loop in line~\ref{dp:loop_ks} of \Cref{alg:dp} takes polynomial time, so the time per iteration $i$ is at most $O\left(\left(\frac{3.2 n}{\ln{n}}\right)^n\right)$. There are $|\m{C}|$ iterations, so the total time complexity is $O\left(\left(\frac{3.2 n}{\ln{n}}\right)^n \cdot |\m{C}|\right)$.
\end{proof}


\ifarxiv
We also implemented other optimizations in the \kernelize algorithm.
See \Cref{subsec:dp_optimizations} for details.
\fi

\begin{table*}[t]
\caption{Our benchmark circuits and their size (number of gates). 
}
\label{tab:benchmarks}
\centering
\small
\begin{tabular}{l|l|rrrrrrrrr}
\toprule
\textbf{Circuit} & \textbf{Description} &\multicolumn{9}{c}{\textbf{Number of qubits}} \\
\textbf{Name} & 
& \textbf{28}  & \textbf{29}  & \textbf{30}  & \textbf{31} & \textbf{32} & \textbf{33} & \textbf{34} & \textbf{35} & \textbf{36} \\
\midrule[2pt]
ae & amplitude estimation & 514 & 547 & 581 & 616 & 652 & 689 & 727 & 766 & 806 \\
dj & Deutsch–Jozsa algorithm & 82 & 85 & 88 & 91 & 94 & 97 & 100 & 103 & 106 \\   
ghz & GHZ state & 28 & 29 & 30 & 31 & 32 & 33 & 34 & 35 & 36 \\                   
graphstate & graph state & 56 & 58 & 60 & 62 & 64 & 66 & 68 & 70 & 72 \\          
ising & Ising model & 302 & 313 & 324 & 335 & 346 & 357 & 368 & 379 & 390 \\                               
qft & quantum Fourier transform & 406 & 435 & 465 & 496 & 528 & 561 & 595 & 630 & 666 \\                   
qpeexact & exact quantum phase estimation & 432 & 463 & 493 & 524 & 559 & 593 & 628 & 664 & 701 \\         
qsvm & quantum support vector machine & 274 & 284 & 294 & 304 & 314 & 324 & 334 & 344 & 354 \\             
su2random & SU2 ansatz with random parameters & 1246 & 1334 & 1425 & 1519 & 1616 & 1716 & 1819 & 1925 & 2034 \\
vqc & variational quantum classifier & 1873 & 1998 & 2127 & 2260 & 2397 & 2538 & 2683 & 2832 & 2985 \\     
wstate & W state & 109 & 113 & 117 & 121 & 125 & 129 & 133 & 137 & 141 \\
\bottomrule
\end{tabular}%
\end{table*}

\subsection{Cost Function in \kernelizeplain}
\label{subsec:dp-cost}
The cost function used in \Cref{eqn:dp_objective} should faithfully represent the running time of each kernel. Inspired by \hyq{}~\cite{chen2021hyquas}, we use two approaches to execute each kernel:
\begin{enumerate}
\item \defn{Fusion}: Fuse all gates in a kernel into a single gate by pre-computing the product of the gate matrices and executing that single gate using \cuq{}~\cite{bayraktar_cuquantum_2023}.
The cost function maps the number of qubits in the kernel to a constant, reflecting the running time of a matrix multiplication of that size.
\item \defn{Shared-memory}: Load the state vector into GPU shared memory in batches and execute the gates one by one instead of fusing them to a single matrix.~\footnote{Shared-memory kernels correspond to \tcd{SHM-GROUPING} in \hyq{}. Similar to \hyq{}, to improve the I/O efficiency of shared-memory kernels, we require the three least significant qubits of the state vector to be in all shared-memory kernels. As a result, each state vector load consists of at least $2^3 = 8$ complex numbers (i.e., 128 bytes as each complex number consists of 2 double-precision floating-point numbers).} The cost function is $\alpha + \sum_{g\in \m{K}} \Call{Cost}{g}$, where $\alpha$ is a constant that corresponds to the time to load a micro-batch of state coefficients into GPU shared memory, and $\Call{Cost}{g}$ is the time to simulate the application of the gate $g$ in the GPU shared memory.
\end{enumerate}
The way to determine these constant values is described in \Cref{subsec:eval-setup}.

In addition to the qubit set and the extensible qubit set, we add a Boolean tag for the type of each kernel (fusion or shared-memory) in \kernelize, and maintain the cost of the kernel during the DP algorithm according to the type. Whenever a new kernel is created, we generate two copies of DP states with the new kernel's type being fusion or shared-memory correspondingly.

\subsection{\Sys}
We built our system \Sys on top of FlexFlow~\cite{jia2019flexflow}, a distributed multi-GPU system for DNN training.
\Sys uses the mapper interface of FlexFlow to distribute the state vector across the DRAM and GPU device memory of a GPU cluster, and uses Legion~\cite{Legion12}, FlexFlow's underlying task-based runtime, to launch simulation tasks on CPUs and GPUs and automatically overlap communications with computations.
Furthermore, \Sys uses the NCCL library~\cite{nccl} to perform inter-GPU communication for sharding. We set the cost factor of the inter-node communication in the circuit staging algorithm $c = 3$ in \Cref{eqn:ilp_objective_informal}.

\Sys is publicly available as an open-source project~\cite{atlas_github} and also in the artifact supporting this paper~\cite{atlas_zenodo}.
\section{Evaluation}

\subsection{Experimental Setup}
\label{subsec:eval-setup}
We use the Perlmutter supercomputer~\cite{perlmutter} to evaluate \Sys. Each compute node is equipped with an AMD EPYC 7763 64-core 128-thread processor, 256 GB DRAM, and four NVIDIA A100-SXM4-40GB GPUs. The nodes are connected with HPE Slingshot 200 Gb/s interconnects. The programs are compiled using GCC 12.3.0, CUDA 12.2, and NCCL 2.19.4.

The cost function used in \kernelize has some constants that require benchmarking on the GPU.
For fusion kernels, we measure their execution time with different numbers of qubits. 
For shared-memory kernels, we measure the run time of an empty shared-memory kernel to estimate the cost of loading a state vector to GPU shared memory, and profile the run times for different types of gates using the GPU shared memory. 

\paragraph{Benchmarks} We collect 11 types of scalable quantum circuits from the MQT Bench~\cite{mqtbench} and NWQBench~\cite{li2021qasmbench}, where we can select the desired number of qubits for each circuit type.
This allows us to conduct a weak-scaling evaluation, where we can compare the performance of \Sys and existing GPU-based quantum circuit simulators using different numbers of GPUs.
\Cref{tab:benchmarks} summarizes the circuits used in our evaluation.

\paragraph{Preprocessing circuits} To speedup simulation, we partition each circuit by (1) splitting it into stages using the \stage algorithm, and (2) kernelizing each stage using the \kernelize algorithm. 
%
This preprocessing needs to be done once per circuit and finishes in 7.2 seconds on average for each circuit in a single thread on an Intel Xeon W-1350 @ 3.30GHz CPU. 
In our evaluation, we use the PuLP library~\cite{dunning_pulp_2011} with the HiGHS solver~\cite{huangfu_parallelizing_2018} for ILP.
The ILP-based \stage algorithm takes 3.3 seconds on average, and \kernelize takes 3.9 seconds on average.
\ifarxiv
A detailed analysis of the running time of \kernelize is included in \Cref{subsec:app-kernelizer}.
\fi

\begin{figure*}
	\centering
	\subfloat[ae]{\includegraphics[width=0.33\textwidth]{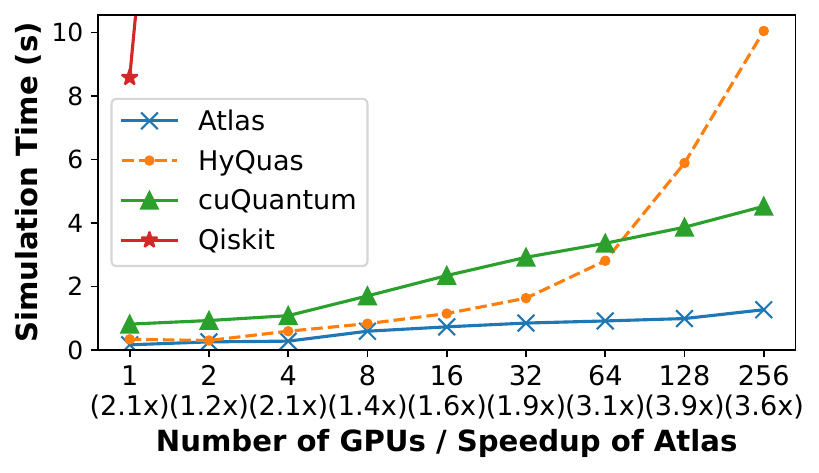}\label{fig:ae}} 
	\subfloat[qft]{\includegraphics[width=0.33\textwidth]{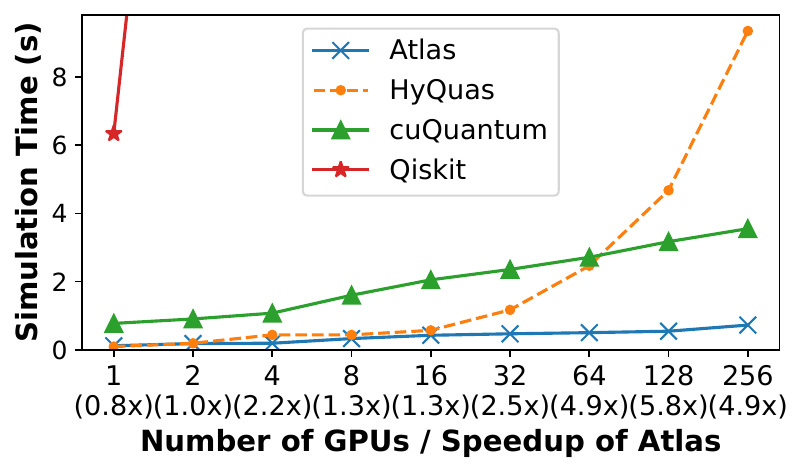}\label{fig:qft}} 
	\subfloat[graphstate]{\includegraphics[width=0.33\textwidth]{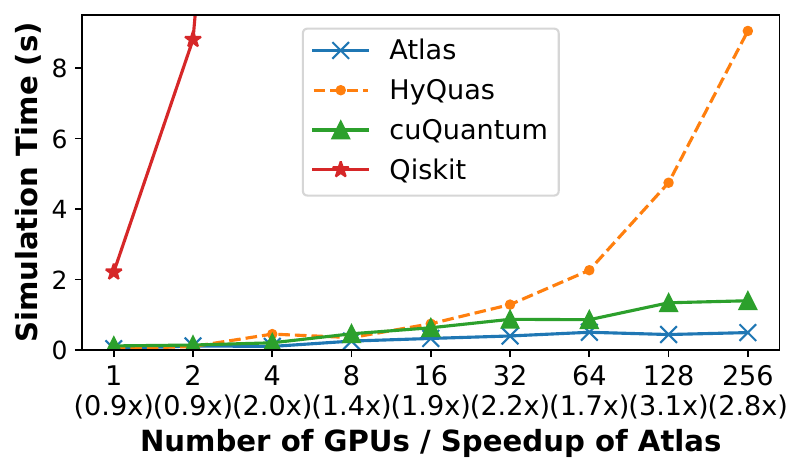}\label{fig:gh}}
 \hfill
 \subfloat[wstate]{\includegraphics[width=0.33\textwidth]{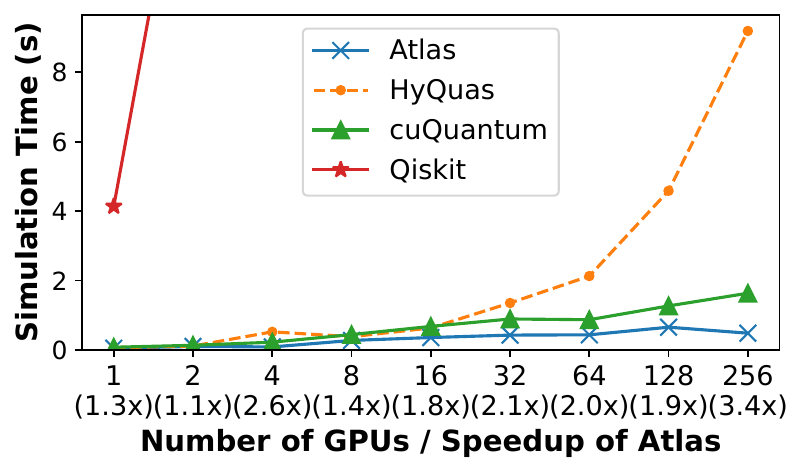}\label{fig:wstate}}
 \subfloat[qsvm]{\includegraphics[width=0.33\textwidth]{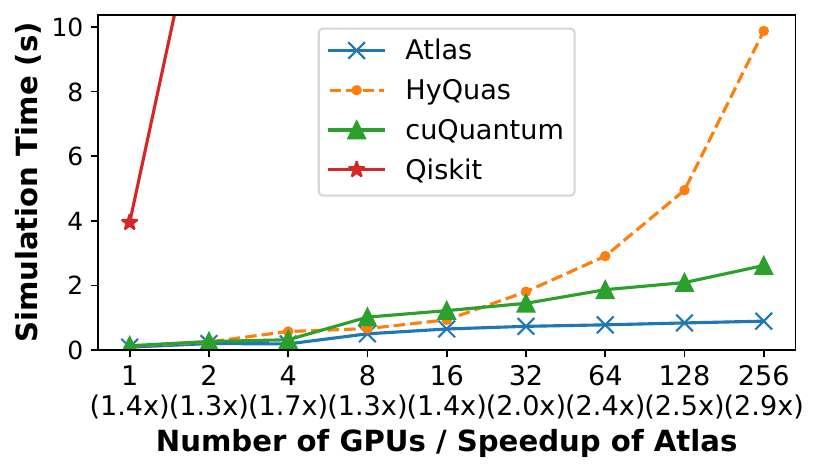}\label{fig:qsvm}}
        \subfloat[ghz]{\includegraphics[width=0.33\textwidth]{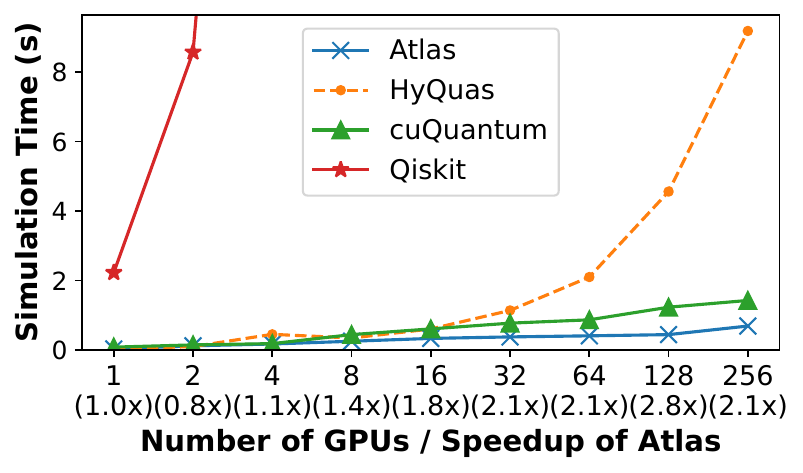}\label{fig:ghz}}
        \hfill
        \subfloat[qpeexact]{\includegraphics[width=0.33\textwidth]{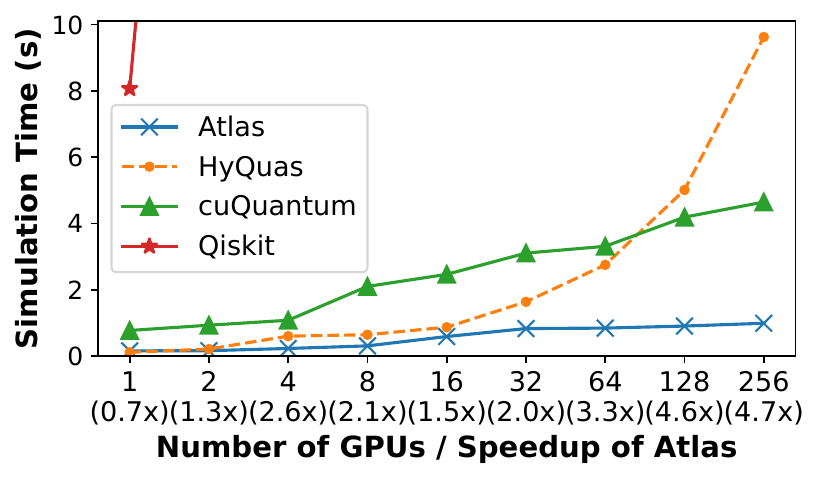}\label{fig:qpeexact}}
        \subfloat[su2random]{\includegraphics[width=0.33\textwidth]{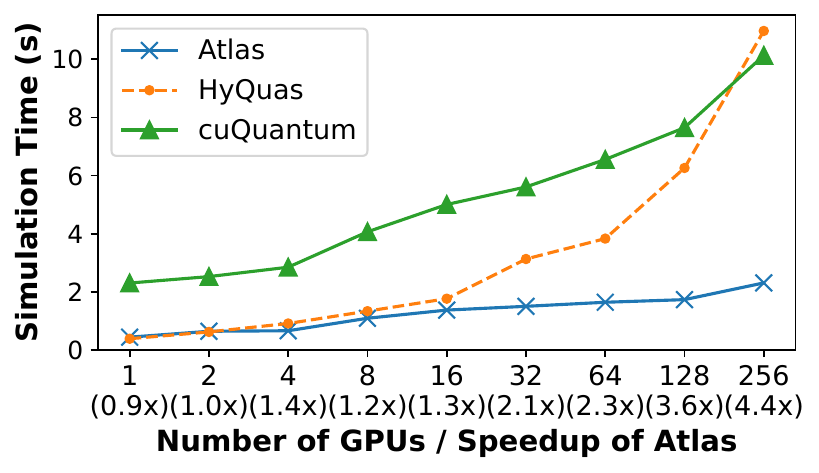}\label{fig:su2random}}
        \subfloat[vqc]{\includegraphics[width=0.33\textwidth]{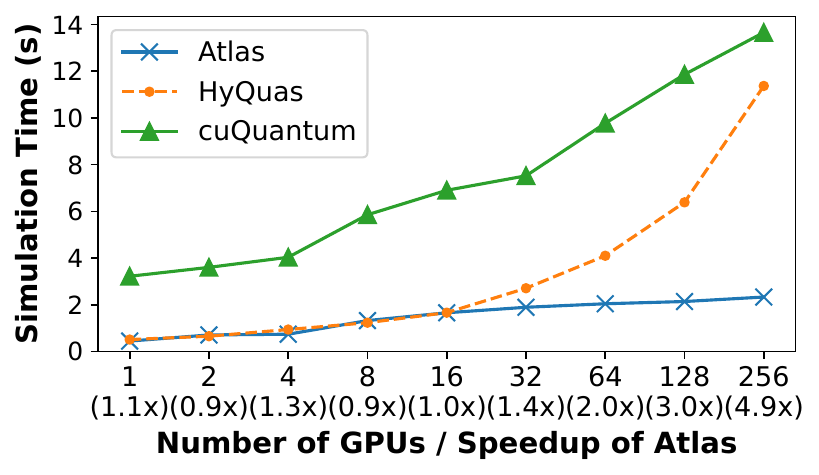}\label{fig:vqc}}
        \hfill
        \subfloat[ising]{\includegraphics[width=0.33\textwidth]{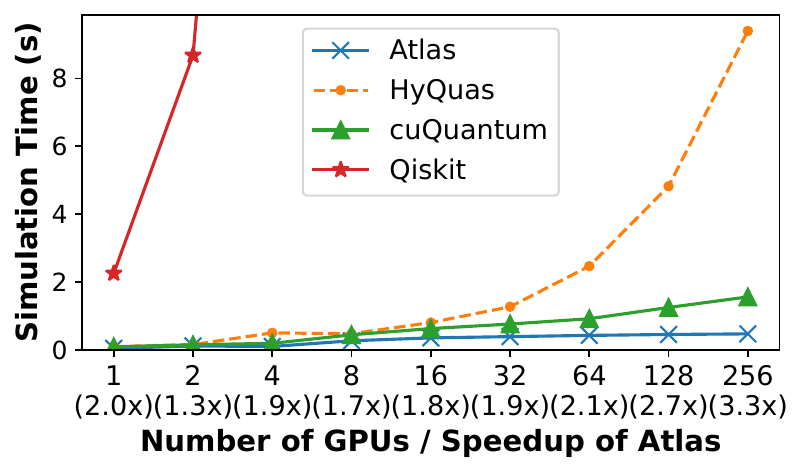}\label{fig:ising}}
        \subfloat[dj]{\includegraphics[width=0.33\textwidth]{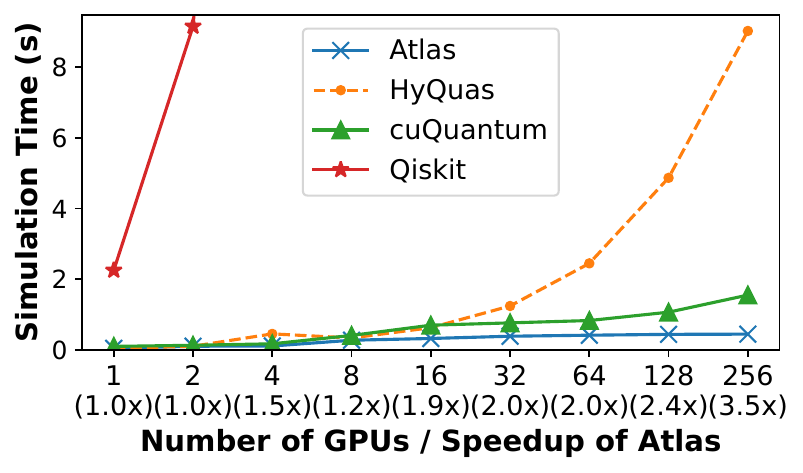}\label{fig:dj}}
        \subfloat[vqc (log scale)]{\includegraphics[width=0.33\textwidth]{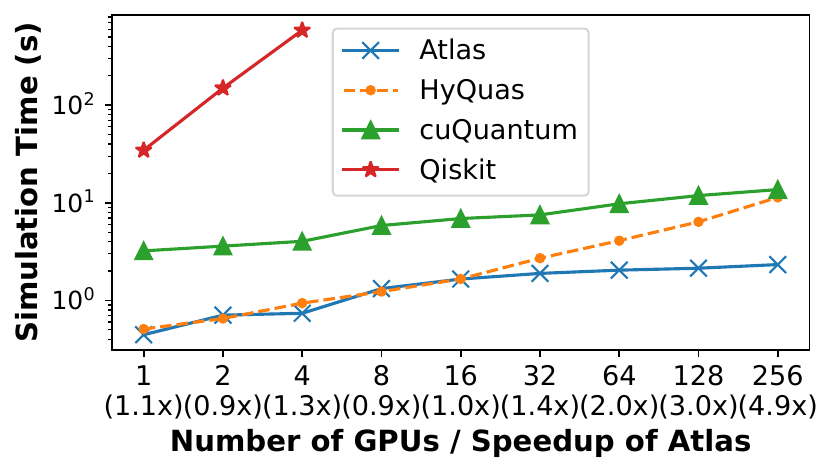}\label{fig:vqc_log}}
	\caption{Weak scaling of \sys{}, \hyq{}, \cuq{}, and \qis{} with 28 local qubits as the number of global qubits increases from 0 (on 1 GPU) to 8 (on 256 GPUs). \qis{} is slow and usually does not fit into our charts.}
	\label{fig:e2e}
\end{figure*}

\subsection{End-to-end Performance}
\label{subsec:end-to-end}
We first compare \Sys against existing distributed GPU-based quantum circuit simulators, including \hyq{}~\cite{chen2021hyquas}, \cuq{}~\cite{bayraktar_cuquantum_2023}, and \qis{}~\cite{qiskit2019}. \uniq{}~\cite{zhang_uniq_2022} is extremely similar to \hyq{} in state-vector-based GPU-based quantum circuit simulation, so we do not compare with \uniq{}.
Both \qis{} and \cuq{} use the \qis{} Python interface as the frontend. \cuq{} uses Aer's cuStateVec integration (\textsf{cusvaer}) to simulate quantum circuits, and \qis{} directly uses its native Aer GPU backend.
All baselines perform state-vector-based simulation and directly store the entire state vector on GPUs.
While \Sys also supports more scalable quantum circuit simulation by offloading the state vector to a much larger CPU DRAM, to conduct a fair comparison, in this experiment, we disable \Syss DRAM offloading and directly store the entire state vector on GPUs. We further evaluate \Syss DRAM offloading performance in \Cref{subsec:scalability}. 




Figure~\ref{fig:e2e} shows the end-to-end simulation performance of the 11 circuit families on up to 256 GPUs (on 64 nodes). 
We use 28 local qubits and increase the number of non-local qubits from 0 (on 1 GPU) to 8 (on 256 GPUs). The number of regional qubits is at most 2 (there are 4 GPUs in each node).
\qis{} is slow, so we only evaluate it on up to 4 GPUs and present a log-scale plot in \Cref{fig:vqc_log}.


\Sys is up to 20.2$\times$ (4.0$\times$ on average) faster than \hyq{}, 7.2$\times$ (3.2$\times$ on average) faster than \cuq{}, and 2,126$\times$ (286$\times$ on average) faster than \qis{} across all types of circuits and possible numbers of GPUs supported by the baselines. 
The speedups over existing systems are achieved by two important optimizations. 
First, \Syss ILP-based circuit staging algorithm can discover a staging strategy that minimizes expensive inter-node communications, allowing \Sys to scale extremely well as the number of GPUs increases.
For example, to scale \texttt{graphstate} from 28 qubits (on 1 GPU) to 36 qubits (on 256 GPUs), \hyq{}' simulation time increases by 267.9$\times$, while the simulation time of \Sys only increases by 13.7$\times$.
Second, the DP-based \kernelize algorithm enables \Sys to achieve high-performance circuit performance on each GPU.

\begin{figure}
    \centering
    \includegraphics[width=0.96\linewidth]{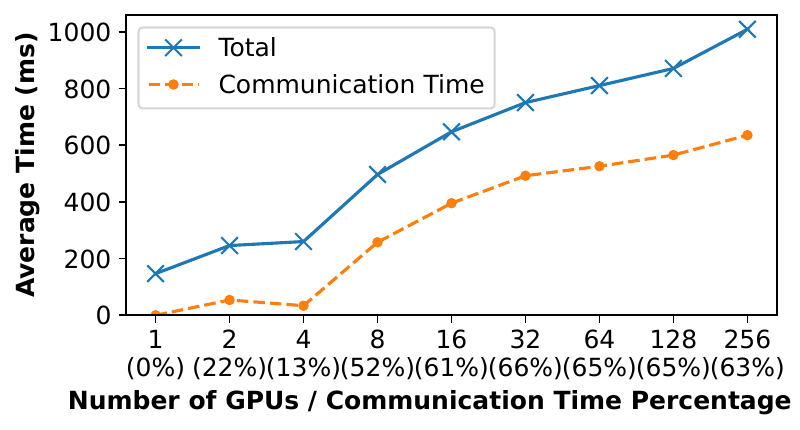}
    \caption{Simulation time breakdown: The average communication time and its percentage in the average total simulation time of \Sys.}
    \label{fig:comm-mean}
\end{figure}

\begin{figure}
\small
\begin{minipage}{0.55\linewidth}
\includegraphics[width=\linewidth]{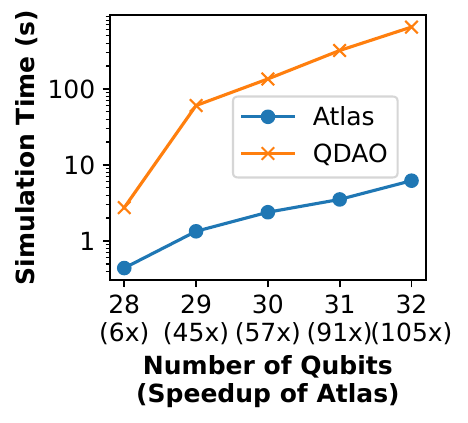}
\caption{
\sys{} outperforms \qdao{} by 61$\times$ on average. 
Log-scale simulation time  (single GPU) with DRAM offloading for \tcd{qft} circuits.}
\label{fig:scalability-qdao}
\end{minipage}
\hfill
\begin{minipage}{0.40333\linewidth}
\includegraphics[width=\linewidth]{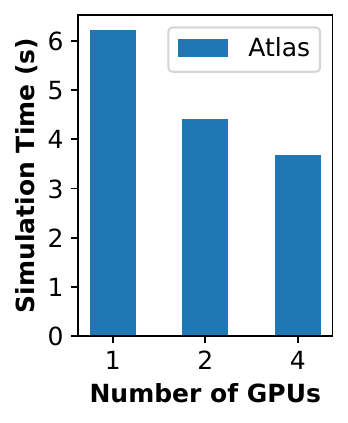}
\caption{DRAM offloading scales. \Sys simulation time of a 32-qubit \tcd{qft} circuit with 1, 2, and 4 GPUs.
}
\label{fig:scalability-atlas}
\end{minipage}
\end{figure}

\begin{figure}
    \centering
\includegraphics[width=0.96\linewidth]{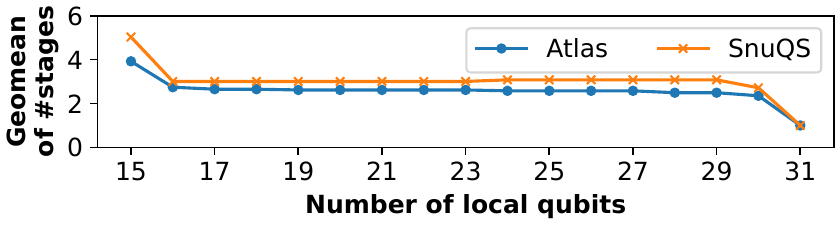}
\caption{Number of stages, \sys{} versus \snq{}: The geometric mean over all our benchmark circuits with 31 qubits.}
\label{fig:ilp_31}
\end{figure}

\begin{figure}
\centering
\includegraphics[width=0.99\linewidth]{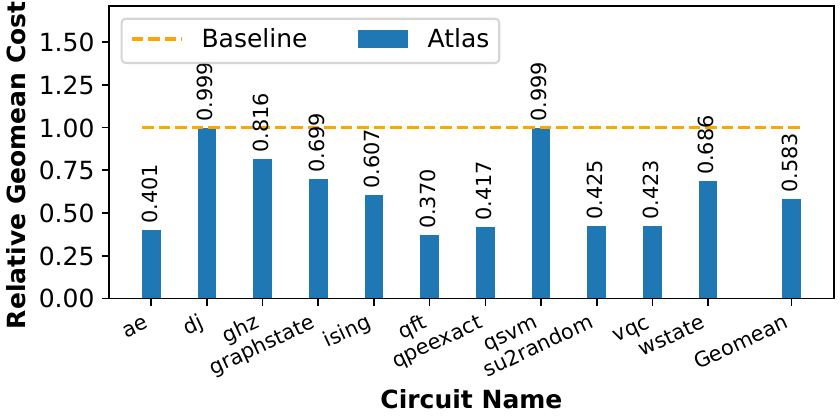}
\caption{
Kernelization effectiveness: The relative geometric mean cost of \kernelizeplain compared to greedy packing up to 5 qubits.
}
\label{fig:dp-circuit-geomean}
\end{figure}

\Cref{fig:comm-mean} breaks down the total simulation time into communication time (including intra-node and inter-node communications) and computation time.
We take the average of the 11 benchmark circuits for each number of GPUs. Inter-node communication is much more expensive than intra-node communication, so the computation is dominant when there is only one node (at most 4 GPUs). When the number of nodes increases, inter-node communication becomes dominant.
We note that the communication time takes a greater portion with 2 GPUs than with 4 GPUs, because each pair of the 4 GPUs in the node is connected. The total bandwidth of each GPU increases by 3$\times$ when the number of GPUs available increases from 2 to 4, so intra-node communication is faster with 4 GPUs than with 2 GPUs.


\subsection{DRAM Offloading}
\label{subsec:scalability}
In addition to improving quantum circuit simulation performance, another key advantage of \Sys over existing systems is its ability to scale to larger circuits that go beyond the GPU memory capacity.
As described in \Cref{sec:atlas}, \Sys does not require that the entire state vector of a quantum circuit be stored in GPU memory and uses DRAM offloading to support larger circuits.
\qdao{}~\cite{zhao2023full} also uses DRAM offloading to support larger circuits and scales to 32 qubits. We compare \Sys with \qdao{} with the \qis{} backend. We use 28 local qubits, and all remaining qubits are regional. For \qdao{}, we set $m=28$, and $t=19$ which runs the fastest.

\Cref{fig:scalability-qdao} shows the simulation performance of \tcd{qft} circuits with different numbers of qubits on a single GPU.
For a 28-qubit circuit, the GPU memory is sufficient, so both \Sys and \qdao are fast; for circuits with 29 to 32 qubits, 
\Sys runs 74$\times$ faster than \qdao{} on average,
and also scales better. The speedups are also achieved by both the \stage and the \kernelize algorithms.

\Cref{fig:scalability-atlas} shows the simulation performance of the 32-qubit \tcd{qft} circuit on 1, 2, and 4 GPUs. \Sys scales across multiple GPUs. The simulation time of \qdao{} stays the same when the number of GPUs increases.




\subsection{Circuit Staging}
We further compare the ILP-based circuit staging algorithm with heuristic-based approaches in existing simulators. 
In particular, we use the heuristics used in \snq{} as a baseline, which greedily selects the qubits with more gates operating on non-local gates to form a stage and uses the number of total gates as a tiebreaker~\cite{park_snuqs_2022}.
Other works did not describe how the heuristics were implemented.

Similar to \Cref{subsec:end-to-end}, we set at most 2 non-local qubits to be regional, and all other non-local qubits to be global.
\Cref{fig:ilp_31} shows the geometric mean number of stages for all 11 circuits with 31 qubits in \Cref{tab:benchmarks}.
Our circuit staging algorithm is guaranteed to return the minimum number of stages by \Cref{thm:ilp-optimal}, and it always outperforms \snq{}' approach. Note that \snq{} may give a worse circuit partition for the same circuit when the number of local qubits increases (from 23 to 24 in \Cref{fig:ilp_31}
), but this is guaranteed not to happen in our ILP-based approach.

\subsection{Circuit Kernelization}
\label{subsec:eval-dp}
We evaluate our \kernelize algorithm and compare it with a baseline that greedily packs gates into fusion kernels of up to 5 qubits, the most cost-efficient kernel size in the cost function used in \Cref{subsec:eval-setup}.

Each circuit family exhibits a pattern, so we take the geometric mean of the cost for 9 circuits with the number of qubits from 28 to 36 for each circuit family. 
\Cref{fig:dp-circuit-geomean} shows that the greedy baseline performs well in \tcd{dj} and \tcd{qsvm} circuits, but it does not generalize to other circuits.
The \kernelize algorithm is able to exploit the pattern for each circuit and find a low-cost kernel sequence accordingly.

\section{Related Work}
\paragraph{Distributed quantum circuit simulators}
Recent work has introduced a number of state-vector-based quantum circuit simulators that
avail parallelism processing elements on the state vector in a data-parallel fashion~\cite{smelyanskiy2016qhipster, zulehner2018advanced,chen201864,wu2019full,chen2021hyquas,park_snuqs_2022,zhao2023full,zhang_uniq_2022}.
Even though the state vector can be trivially partitioned, applying quantum gates on it can easily result in a large amount of communication (e.g., between memory and processors and between different nodes).
Researchers have therefore proposed circuit partitioning techniques
that coalesce computations that operate on a smaller subset of the
state space as a single partition so as to reduce communication
costs~\cite{efthymiou_qibo_2020,park_snuqs_2022,chen2021hyquas,suzuki2021qulacs}.
An important limitation of all this prior work is that they use
heuristic techniques to partition the circuits, leading to suboptimal
behavior.
We formulate this circuit partitioning problem as an
integer linear programming (ILP) problem, which can be solved by
applying existing solvers optimally.

\paragraph{Quantum gate fusion} Prior work has developed gate
fusion techniques that use various heuristics to fuse multiple gates that operate
on nearby qubits into a single larger gate which can then be applied at
once~\cite{qiskit2019,suzuki2021qulacs,chen2021hyquas,westrick2024grafeyn}.
For example, Qulacs either leaves it to the user or applies the ``light'' or ``heavy'' approaches to merging gates~\cite{suzuki2021qulacs}.
In this paper, we formulate the circuit kernelization problem of grouping gates into kernels and propose a dynamic programming algorithm to systematically solve this problem.

\paragraph{Domain-specific quantum circuit simulators}
This paper focuses on the general quantum circuit simulation problem and makes no assumption about input circuits.
There has been significant research on developing simulators
for specific classes of quantum
circuits~\cite{jozsa2003,aharonov2006quantum,gottesman1998heisenberg,markov2008simulating,zhao2022q,zhang_uniq_2022,lykov_fast_2023,lykov_tensor_2022}.

In addition, our work considers an idealized quantum computing environment
and does not attempt to simulate errors experienced by modern, NISQ-era
quantum computers.
There has also been work on developing error-aware quantum circuit simulators~\cite{trieu2010large,li2020density}.
Error simulation is more expensive, because of the need to track errors.
\section{Limitation}
Although our techniques can improve run time, the running time of algorithms \stage{} and \kernelize{} may depend on the circuit structure, the circuit size, and the ILP solver used.
For example, \Cref{thm:dp-complexity} establishes an exponential upper bound in the number of qubits for the running time of \kernelize.
But we do not know if this upper bound is tight: our algorithms only take a few seconds on average to preprocess each circuit in our evaluation.
We believe that improving the scalability of algorithms \stage and \kernelize is an important problem and will need to be revisited when quantum circuit simulators are able to handle larger circuits.
Future work could narrow the gap between theoretical time complexity and practice and explore the trade-off between preprocessing and simulation.

\section{Conclusion}


In this paper, we propose
an ILP-based algorithm \stage{} that minimizes the communication cost between GPUs, and a dynamic programming algorithm, \kernelize{}, that ensures efficient kernelization of the GPU work.
We then present an implementation of a distributed, multi-GPU quantum circuit simulator, called \Sys, that realizes the proposed algorithms.
%
%
Previous work also partitioned the circuit for improved communication and kernelization, but relied on heuristics to determine the partitioning.
We show that it is possible to achieve substantial improvements by using provable algorithms that can partition the circuit 
hierarchically.

\section*{Acknowledgement}
We thank the anonymous SC reviewers for their feedback on this work. This research is partially supported by NSF awards CNS-2147909, CNS-2211882, CNS-2239351, CCF-1901381,
CCF-2115104, CCF-2119352, and CCF-2107241 and research awards from Amazon, Cisco, Google, Meta, Oracle, Qualcomm, and Samsung.
This research used resources of the National Energy Research Scientific Computing Center (NERSC), a U.S. Department of Energy Office of Science User Facility located at Lawrence Berkeley National Laboratory, operated under Contract No. DE-AC02-05CH11231 using NERSC award DDR-ERCAP0023403.


\bibliographystyle{IEEEtran}
\bibliography{new}

\ifarxiv
\clearpage
\appendix
\section{Appendix}

\subsection{An Alternate Algorithm for Kernelization}
\label{app:dp-simple}
\begin{algorithm}
{
\small
\begin{algorithmic}[1]
\Function{$\Call{OrderedKernelize}{}$}{$\m{C}$}

\State {\bf Input:} A quantum circuit $\m{C}$ represented as a sequence of gates.
\State {\bf Output:} A sequence of kernels.
\State //  $DP[i]$ stores the minimum cost to kernelize the first $i$ gates and the corresponding kernel sequence.

\State $DP[0] = (0, [])$
\For {$i = 0$ \textbf{to} $|\m{C}|-1$} 
\State $DP[i+1] = \min_{0 \le j \le i}\{DP[j] + \Call{Cost}{\m{C}[j, \ldots,i]}\}$\label{line:dp-simple-add}
\EndFor
\State \Return $DP[|\m{C}|].kernels$
\EndFunction
\end{algorithmic}
}
\caption{The \dporderedplain Algorithm.
The ``+'' operator in line~\ref{line:dp-simple-add} means adding the cost and appending the kernel to the sequence.
}
\label{alg:dp-simple}
\end{algorithm}

We present \Cref{alg:dp-simple}, an alternate dynamic programming (DP) algorithm that also solves the
optimal circuit kernelization problem, considering only kernels with contiguous gate segments.
The algorithm, called \dpordered{}, takes the circuit
representation as a sequence of gates and outputs a sequence of kernels.
The variable \tcd{DP} maps a gate (position) in the circuit to a pair
of the total cost for the circuit prefix up to that gate (excluding
the gate) and the sequence of kernels for that prefix.
The algorithm starts with a total cost of zero and an empty kernel
sequence.
It then considers each possible gate $\m{C}[i]$ in increasing order of $i$ and
computes for that gate the optimal cost.
To compute the optimal cost at a position, \dpordered{} examines the cost of all possible kernels ending at that position and takes the minimum.
Because for each gate, the algorithm considers all possible kernels
ending at that gate, its run-time is asymptotically bounded by the
square of the length of the circuit, i.e., $O(|\m{C}|^2)$.

Although \dpordered{} is easier to implement, it has a significant limitation: It only considers a specific sequential ordering of the circuit.
It is a priori difficult to determine an ordering that yields
the lowest cost. For example, 
the optimal kernel sequence for \Cref{fig:dp_with_optimization} might be $\{g_0, g_8, g_9\}, \{g_1, g_2, g_4\}, \{g_3\}, \{g_5, g_6, g_7\}$. \dpordered{} cannot find this kernel sequence if it is not provided with the corresponding ordering of the circuit. 
We show that the \kernelize algorithm is better than \dpordered{} both theoretically (\Cref{thm:dp-optimal}) and empirically (\Cref{fig:dp-pruning}).

\begin{theorem}[The \kernelize algorithm is at least as good as \dpordered{}]
\Cref{alg:dp} always returns a kernel sequence with a total execution cost no more than \Cref{alg:dp-simple}, which is optimal for \Cref{problem:kernelization}.
\label{thm:dp-optimal}
\end{theorem}

\begin{proof}
Let $\m{K}_0,\dots,\m{K}_{b-1}$ be the kernel sequence that minimizes the total execution cost such that each kernel consists of a contiguous segment of the gate sequence $\m{C}$. It suffices to prove that \Sys kernelization algorithm returns a kernel sequence with a total execution cost no more than \Cref{eqn:dp_objective}.

For any $0 \le i < |\m{C}|$, suppose $\m{C}[i] \in \m{K}_t$ ($0 \le t < b$).
We prove by induction on $i$ that for any $0 \le i < |\m{C}|$, there is a kernel set $\KS{}_{i+1}$ such that
\begin{enumerate}
    \item $\{\kr{\m{K}_t}{(i+1)}\} \subseteq\KS{}_{i+1}\subseteq \{\m{K}_0, \dots, \m{K}_{t-1}, \kr{\m{K}_t}{(i+1)}\}$
    \item and after the iteration $i$ (0-indexed) of the loop in line~\ref{dp:loop_i} of \Cref{alg:dp}, $DP[i+1, \KS{}_{i+1}]$ has a cost no more than $\Call{Cost}{\{\m{K}_0, \dots, \m{K}_{t-1}, \kr{\m{K}_t}{(i+1)}\}\setminus\KS{}_{i+1}}$.
\end{enumerate}

When $i=0$, line~\ref{dp:singleton-kernel} transits from $DP[0,\emptyset] = (0, [])$ to $DP[1,\{\m{C}[0]\}] = (0, [])$ where $\KS{}_1=\KS{}=\{\m{C}[0]\}, \KS{}'=\tilde{\KS{}}=\emptyset$, so the induction hypothesis holds.

When $i>0$, $\kr{\m{K}_t}{(i+1)}$ satisfies \Cref{dp-constraint} by \Cref{thm:contiguous}. If $|\kr{\m{K}_t}{(i+1)}| > 1$, we apply the induction hypothesis on $i-1$ and it holds for $i$ because line~\ref{dp:add-directly} sets $DP[i+1,\KS{}_{i+1}]$ to $DP[i,\KS{}_i]$ in the iteration where $\KS{}_{i+1}=\KS{}=\KS{}_i \setminus \{\kr{\m{K}_t}{i}\} \cup \{\kr{\m{K}_t}{(i+1)}\}$. If $|\kr{\m{K}_t}{(i+1)}| = 1$, consider the transition in line~\ref{dp:singleton-kernel} in the iteration where $\KS{}_{i+1}=\KS{}=\{\kr{\m{K}_t}{(i+1)}\}$. After taking the minimum when $\KS{}' = \KS{}_i$,
\begin{equation*}
\begin{aligned}
&DP[i+1, \KS{}_{i+1}].cost \\
\le\ & DP[i, \KS{}_i].cost + \Call{Cost}{\KS{}_i} \\
\le\ & \Call{Cost}{\{\m{K}_0, \dots, \m{K}_{t-2}, \kr{\m{K}_{t-1}}{i}\}\setminus\KS{}_i} + \Call{Cost}{\KS{}_i} \\
=\ & \Call{Cost}{\{\m{K}_0, \dots, \m{K}_{t-1}\}} \\
=\ & \Call{Cost}{\{\m{K}_0, \dots, \m{K}_{t-1}, \kr{\m{K}_t}{(i+1)}\}\setminus \KS{}_{i+1}}.
\end{aligned}
\end{equation*}

So after the last iteration of the loop in line~\ref{dp:loop_i}, there is a kernel set $\KS{}_{|\m{C}|}$ such that $DP[|\m{C}|, \KS{}_{|\m{C}|}]$ has a cost no more than $\Call{Cost}{\{\m{K}_0, \dots, \m{K}_{b-1}\}\setminus\KS{}_{|\m{C}|}}$, which evaluates to 
$$
\Call{Cost}{\{\m{K}_0, \dots, \m{K}_{b-1}\}}=\sum_{i=0}^{b-1} \Call{Cost}{\m{K}_i}
$$
and is taken minimum with the returned kernel sequence on line~\ref{dp:compute_f_best}.
\end{proof}




\subsection{Optimizations to the \kernelizeplain 
Algorithm}
\label{subsec:dp_optimizations}
This section describes the optimizations when implementing \Cref{alg:dp} with \Cref{alg:extensible}.

\paragraph{\Kwins{} qubits}
Recall that the \Sys circuit staging algorithm allows mapping the insular qubits of a gate to regional and global qubits. 
Before loading a batch of states to GPU, the values of all global and regional qubits for the batch are known, which are leveraged by \Sys to simplify the computation.
Specifically, for gates with \kwins{} non-local qubits, \Sys retrieves the non-local qubits' values and converts these gates into local gates with fewer qubits.

We can do more along this line for shared-memory kernels. A physical qubit is \defn{active} for a shared-memory kernel if its states are loaded into the shared memory in a single micro-batch. 
For \emph{non-active} local qubits, \Sys applies a similar optimization to gates with \kwins{} non-active local qubits that converts these gates into smaller gates with only active qubits.
Because we can swap two adjacent gates if all qubits they share are \kwins{} to both gates, we can slightly lift the \Cref{dp-constraint} by ignoring any qubit in line~\ref{extensible:intersect} of \Cref{alg:extensible} that is \kwins{} to $\m{C}[i]$ and all gates in $\m{K}'$. In order to know this, we additionally maintain an active qubit set for each shared-memory kernel $\m{K}'$ for the \kwnonins{} qubits of all gates in $\m{K}'$.

Furthermore, we additionally maintain an extensible \kwins{} qubit set that is maintained in the same way as the extensible qubit set in line~\ref{extensible:remove-qubits} of \Cref{alg:extensible} but initialized as 
$\Call{ExtensibleInsularQubits}{\m{K}',i+1} = \Call{AllQubits}{} \setminus \Call{Qubits}{\m{C}[i]}$ in line~\ref{extensible:qubits-k} of \Cref{alg:extensible}. We can then further lift the \Cref{dp-constraint} by allowing to add a gate $\m{C}[i]$ to a shared-memory kernel $\m{K}'$ as long as its \kwins{} qubits are in the extensible \kwins{} qubit set and its \kwnonins{} qubits are in the extensible qubit set.

\begin{figure}
    \centering
    \includegraphics[scale=0.5]{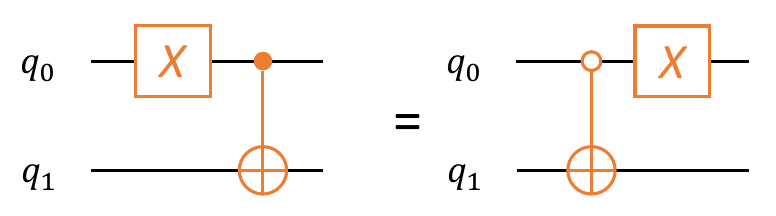}
    \caption{If we swap an $X$ gate with an adjacent $CX$ gate, we need to flip the control state of the new $CX$ gate to $\ket{0}$.}
    \label{fig:swap}
\end{figure}

After the DP algorithm, we reconstruct the topological order to see which gates are swapped. If one of the gates being swapped is anti-diagonal (e.g., the $X$ gate) and the other gate is a control gate, \Sys flips the control state of that gate, as shown in \Cref{fig:swap}.

This optimization opens up more ways to reorder the circuit sequence.

\paragraph{Optimizing DP transition}
When a new gate $\m{C}[i]$ can be subsumed by an existing kernel $\m{K}$ or vice versa, i.e., $\Call{Qubits}{\m{C}[i]} \subseteq \Call{Qubits}{\m{K}}$ or $\Call{Qubits}{\m{K}} \subseteq \Call{Qubits}{\m{C}[i]}$ (and $\Call{Qubits}{\m{C}[i]} \subseteq \Call{\extq}{\m{K}}$ by \Cref{dp-constraint}),
\Sys directly adds $\m{C}[i]$ into $\m{K}$. In other words, \Sys will not consider any $\KS{}$ where $\m{C}[i] \in \m{K}' \in \KS{}$ for any $\m{K}' \neq \m{K}$.
For example, in \Cref{fig:dp_with_optimization}, \Sys will directly add $g_7$ into $\m{K}_4$ without considering other options.
This optimization greatly reduces the amount of DP transitions \Sys must explore.

\paragraph{Deferring adding gates to extensible kernels}
To further reduce the number of DP states, we keep gates as single-gate kernels when their extensible qubit sets are all qubits instead of trying to add them to each of the kernels. Upon running line~\ref{extensible:qubits-k} of \Cref{alg:extensible} for $\m{K}'$, we search for all other kernels $\m{K}''$ with $\Call{\extq}{\m{K}'', i}=\Call{AllQubits}{}$ to be merged together with $\m{K}'$ at that time.

For example, after adding $g_7$ into $\m{K}_4$ in \Cref{fig:dp_with_optimization}, we create a singleton kernel $\m{K}_5$ for $g_8$. In the next iteration, we find that both $\m{K}_1$ and $\m{K}_5$ are subsumed by $g_9$, so we pick one arbitrarily to add $g_9$ into it. Suppose we add $g_9$ into $\m{K}_5$, then we will run line~\ref{extensible:qubits-k} of \Cref{alg:extensible} for $\m{K}_1$. Before running this line, we search for all other kernels with the extensible qubit set being all qubits, i.e., $\m{K}_4$ and $\m{K}_5$. So we will either leave $\m{K}_1$ alone or merge it with $\m{K}_4$ or $\m{K}_5$, resulting in different $\KS{}$'s. In this particular example, merging $\m{K}_1$ with $\m{K}_5$ will be the best option.

\paragraph{Attaching single-qubit gates}
Our preliminary experiments showed that the single-qubit gates of a quantum circuit could easily result in an exponential blow-up in the number of DP states because of the independence between single-qubit gates applied to different qubits. To address this issue, we attach each single-qubit gate to an adjacent multi-qubit gate. If the single-qubit gate is \kwnonins{}, we mark the corresponding qubit \kwnonins{} in the attached gate. For single-qubit gates that are only adjacent to other single-qubit gates, we attach them together to the closest multi-qubit gate.

\paragraph{Post-processing}
After the DP algorithm, \Sys needs to merge some kernels together in $\KS{}$ before taking the cost as a post-processing step. 
Computing the minimum cost for merging kernels is NP-hard because the bin packing problem can be reduced to this problem.
\Sys uses a simple approximation algorithm that greedily packs the kernels in $\KS{}$, running in quasilinear time: (1) for fusion kernels, \Sys packs them towards the most cost-efficient density; (2) for shared-memory kernels, \Sys packs them towards the maximum shared-memory kernel size.

\paragraph{Pruning DP states}
Another key optimization to reduce the DP states is {\em pruning}. A simple pruning based on the cost of the DP states would be suboptimal because we have not considered the cost of the kernels in $\KS{}$ which may contribute a lot.
During the DP, whenever the number of states at the same position $i$ in the circuit reaches a user-defined threshold $T$ at the beginning of an iteration (i.e., Line~\ref{dp:loop_i} in \Cref{alg:dp}), \Sys performs the post-processing step introduced above at the current position for all DP states and keeps $\frac{T}{2}$ DP states with the lowest cost after post-processing.
This is the only optimization that may worsen the results. 
A study of the effect of $T$ is included in \Cref{subsec:app-kernelizer}.
We set $T = 500$ in our evaluation, which achieves good performance.

\subsection{Detailed Results}

\subsubsection{Circuit Staging}
We also evaluated \Sys staging algorithm for circuits with more qubits, and the results are similar to \Cref{fig:ilp_31}. We present the 42-qubit results in \Cref{fig:ilp_42}.
\begin{figure}[h]
\includegraphics[width=0.99\linewidth]{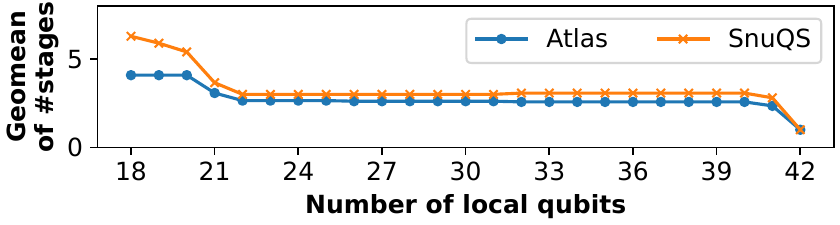}
\caption{Number of stages, \sys{} versus \snq{}: The geometric mean on 11 benchmark circuits with 42 qubits.}
\label{fig:ilp_42}
\end{figure}

\subsubsection{Circuit Kernelization}
\label{subsec:app-kernelizer}
\paragraph{Results for each individual circuit}
\Cref{fig:dp_ae,fig:dp_dj,fig:dp_ghz,fig:dp_graphstate,fig:dp_ising,fig:dp_qft,fig:dp_qpeexact,fig:dp_qsvm,fig:dp_su2random,fig:dp_vqc,fig:dp_wstate} show the resulting cost of \Syss kernelizer for each circuit compared against a baseline greedily packing gates into fusion kernels of up to 5 qubits. ``\Sys-Naive'' denotes \dpordered{}.
\Cref{fig:dp_time_ae,fig:dp_time_dj,fig:dp_time_ghz,fig:dp_time_graphstate,fig:dp_time_ising,fig:dp_time_qft,fig:dp_time_qpeexact,fig:dp_time_qsvm,fig:dp_time_su2random,fig:dp_time_vqc,fig:dp_time_wstate} show the preprocessing time of them.

\paragraph{Case study on a circuit with many gates}
We additionally present a case study where the number of gates is far larger than the number of qubits. We take the \tcd{hhl} circuit in NWQBench~\cite{li2021qasmbench} with 4, 7, 9, and 10 qubits, and pad each circuit to 28 qubits to restrict the kernelization algorithms to use GPU instead of computing the matrix for a 10-qubit gate directly. The number of gates is given in \Cref{tab:hhl_gates}.
\begin{table}[ht]
\caption{\label{tab:dp-geomean}The number of gates in the \tcd{hhl} circuit.}
\label{tab:hhl_gates}
\centering
\small
\begin{tabular}{r|r}
\toprule
\textbf{Number of qubits} & \textbf{Number of gates} \\
\midrule
4 & 80 \\
7 & 689 \\
9 & 91,968 \\
10 & 186,795 \\
\bottomrule
\end{tabular}%
\end{table}

\Cref{fig:dp_hhl} shows the resulting cost and \Cref{fig:dp_time_hhl} shows the preprocessing time of the kernelization algorithms. \kernelize runs in linear time to the number of gates, and it is faster than \dpordered{} on each of these circuits while achieving the same resulting cost.

\begin{figure*}[bt]
\centering
\includegraphics[width=0.99\linewidth]{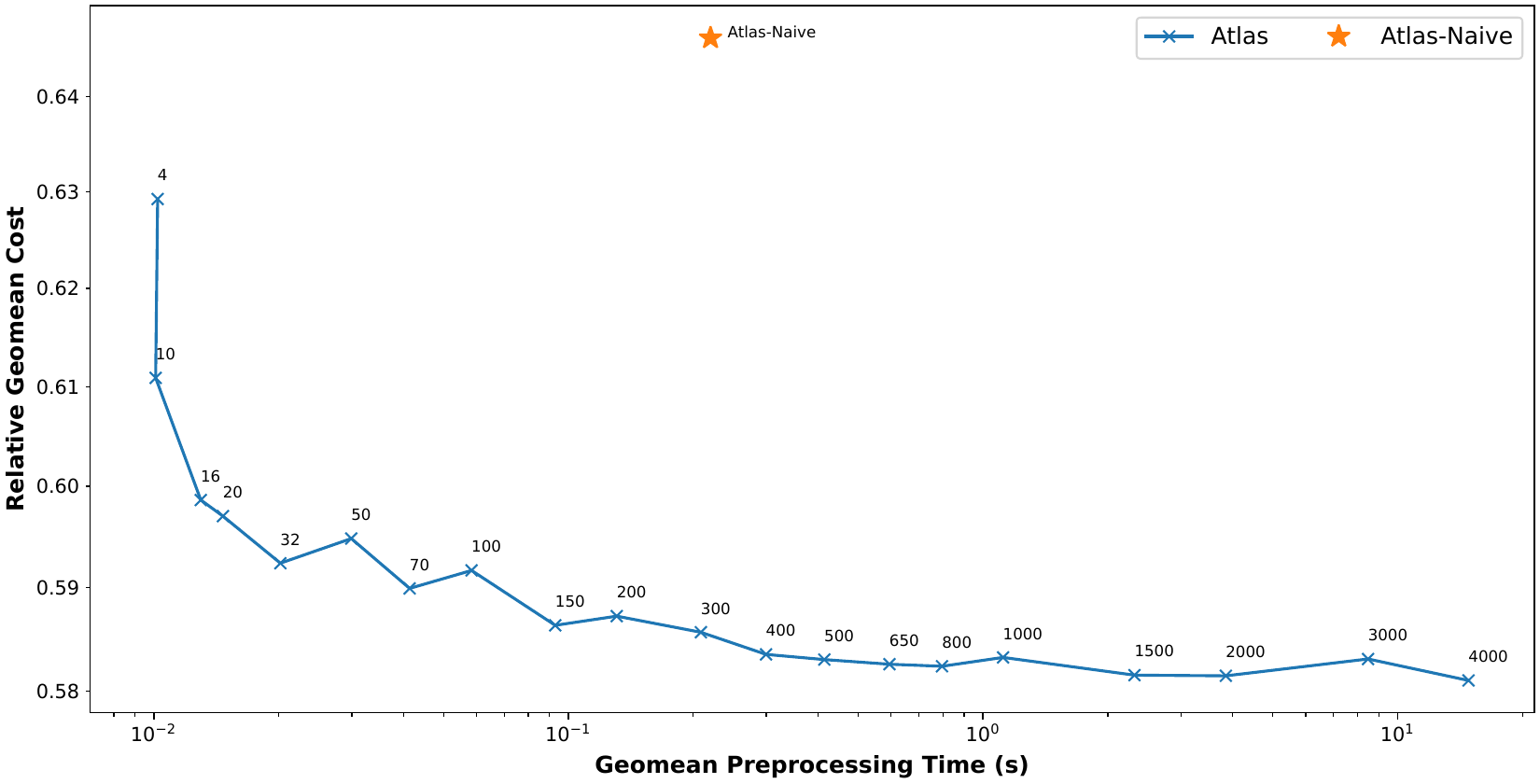}
\caption{The relative geometric mean cost among all 99 circuits in \Cref{tab:benchmarks} of \kernelizeplain compared to greedy packing up to 5 qubits, and the running time of the kernelization algorithm when the pruning threshold $T$ varies. Each data point is annotated with the pruning threshold $T$. ``\Sys'' denotes \kernelizeplain. ``\Sys-Naive'' denotes \dporderedplain.
}
\label{fig:dp-pruning}
\end{figure*}

\paragraph{Pruning threshold and running time of \kernelizeplain}
\label{app:dp-pruning}
\kernelize implements a pruning threshold $T$ in \Cref{subsec:dp_optimizations}.
As shown in \Cref{fig:dp-pruning}, the running time of \kernelize increases and the total execution cost decreases when $T$ increases. The cost reduction benefit diminishes quickly and the running time of \kernelize grows exponentially when $T$ increases exponentially; even when $T=4$, \kernelize achieves a lower total execution cost than \dpordered{} and also preprocesses faster. We choose $T=500$ in our evaluation as the geometric mean cost starts to flatten and the running time of \kernelize (3.9 seconds on average) is on the same order of magnitude as the ILP-based circuit staging algorithm (3.3 seconds on average). Because some large circuits take a much longer time to preprocess than small circuits, the average preprocessing time is longer than the geometric mean. $T=4,000$ or even higher would give us better results, but that would require a much longer circuit preprocessing time.

\begin{figure}[H]
\centering
\includegraphics[width=0.97\linewidth]{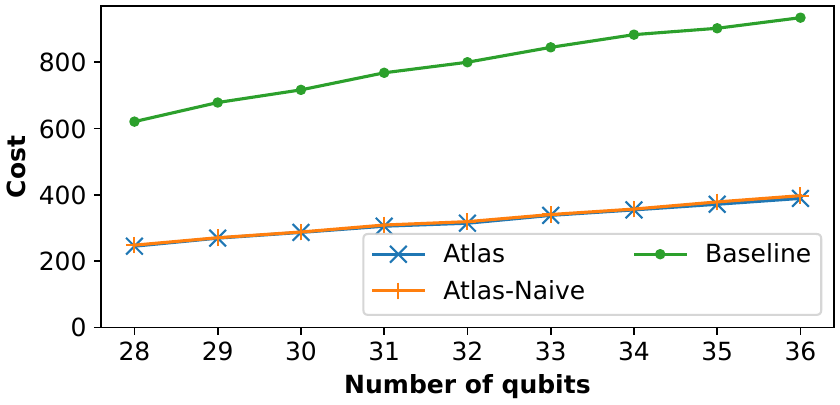}
\caption{The total execution cost of different kernelization algorithms on the circuit \tcd{ae}.}
\label{fig:dp_ae}
\end{figure}

\begin{figure}[H]
\centering
\includegraphics[width=0.97\linewidth]{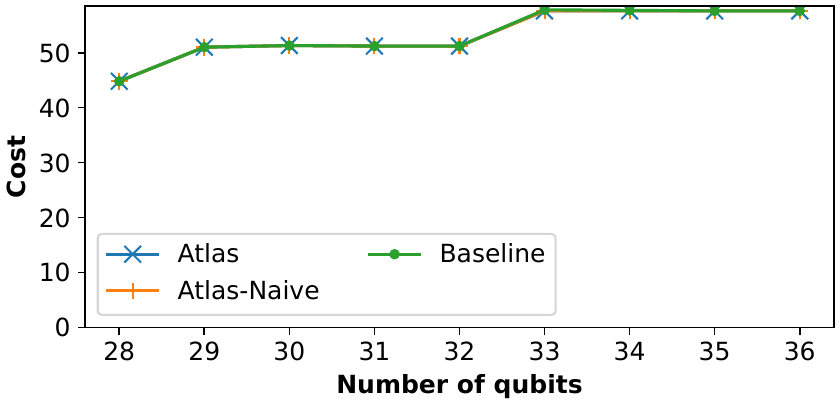}
\caption{The total execution cost of different kernelization algorithms on the circuit \tcd{dj}.}
\label{fig:dp_dj}
\end{figure}

\begin{figure}[H]
\centering
\includegraphics[width=0.97\linewidth]{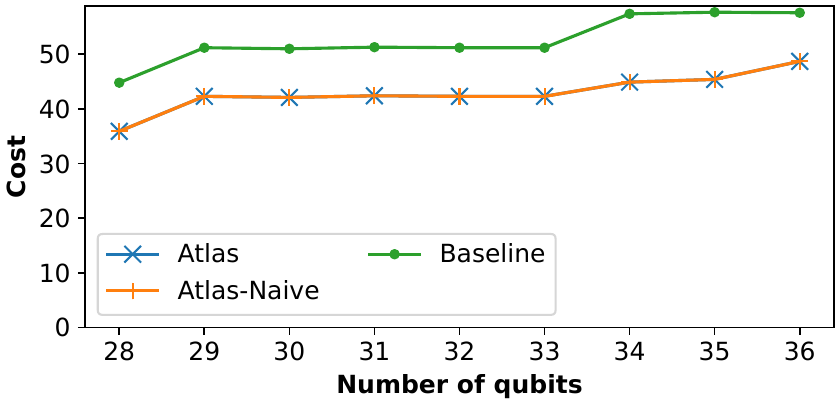}
\caption{The total execution cost of different kernelization algorithms on the circuit \tcd{ghz}.}
\label{fig:dp_ghz}
\end{figure}

\begin{figure}[H]
\centering
\includegraphics[width=0.97\linewidth]{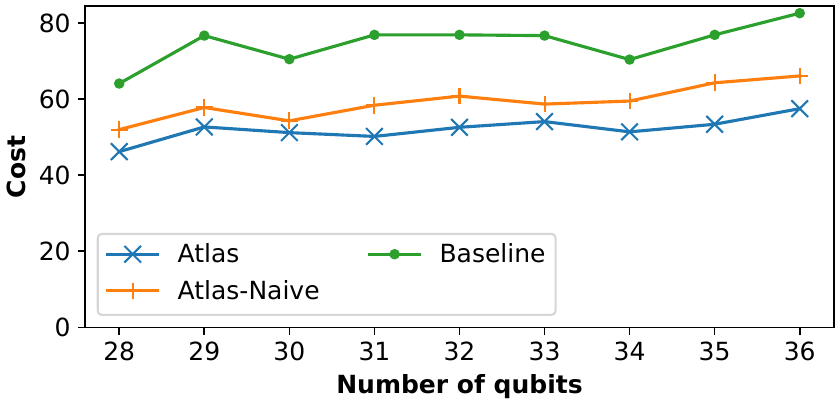}
\caption{The total execution cost of different kernelization algorithms on the circuit \tcd{graphstate}.}
\label{fig:dp_graphstate}
\end{figure}

\begin{figure}[H]
\centering
\includegraphics[width=0.97\linewidth]{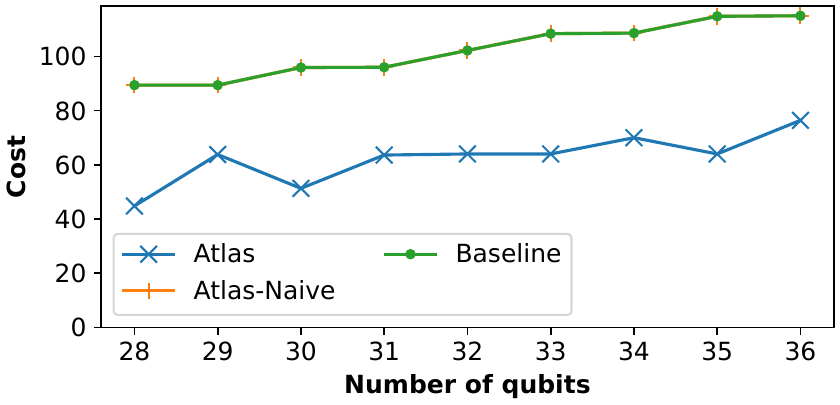}
\caption{The total execution cost of different kernelization algorithms on the circuit \tcd{ising}.}
\label{fig:dp_ising}
\end{figure}

\begin{figure}[H]
\centering
\includegraphics[width=0.97\linewidth]{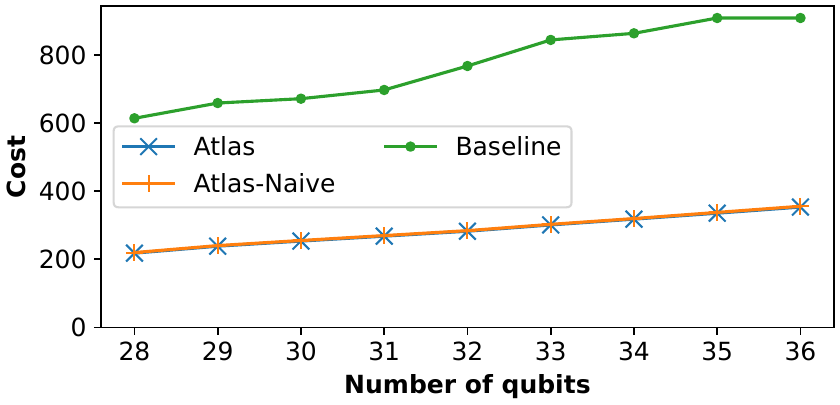}
\caption{The total execution cost of different kernelization algorithms on the circuit \tcd{qft}.}
\label{fig:dp_qft}
\end{figure}

\begin{figure}[H]
\centering
\includegraphics[width=0.97\linewidth]{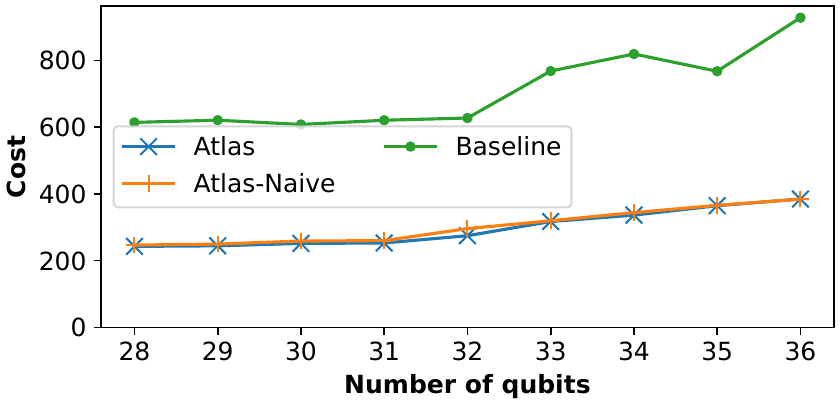}
\caption{The total execution cost of different kernelization algorithms on the circuit \tcd{qpeexact}.}
\label{fig:dp_qpeexact}
\end{figure}

\begin{figure}[H]
\centering
\includegraphics[width=0.97\linewidth]{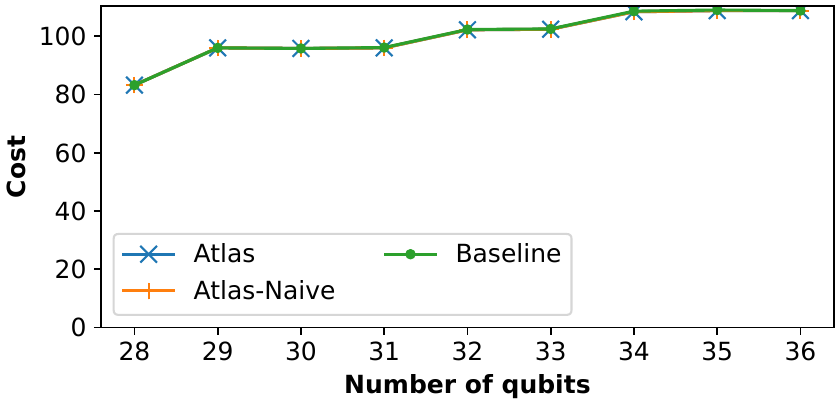}
\caption{The total execution cost of different kernelization algorithms on the circuit \tcd{qsvm}.}
\label{fig:dp_qsvm}
\end{figure}

\begin{figure}[H]
\centering
\includegraphics[width=0.97\linewidth]{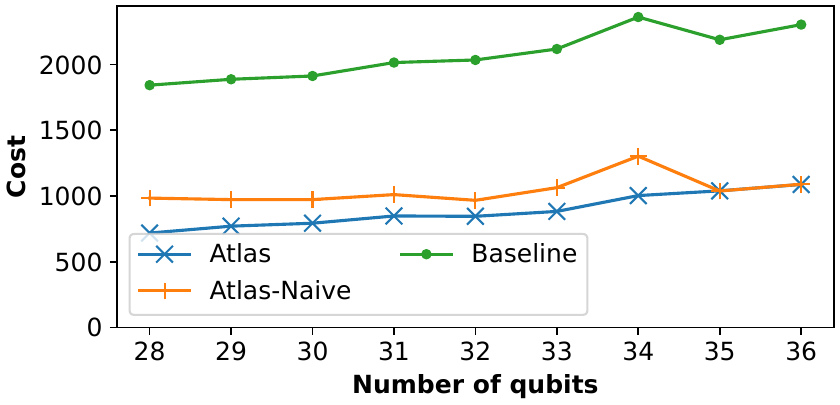}
\caption{The total execution cost of different kernelization algorithms on the circuit \tcd{su2random}.}
\label{fig:dp_su2random}
\end{figure}

\begin{figure}[H]
\centering
\includegraphics[width=0.97\linewidth]{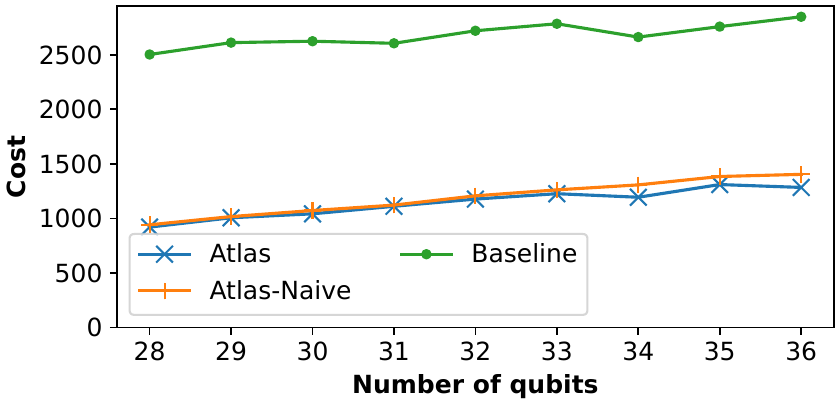}
\caption{The total execution cost of different kernelization algorithms on the circuit \tcd{vqc}.}
\label{fig:dp_vqc}
\end{figure}

\begin{figure}[H]
\centering
\includegraphics[width=0.97\linewidth]{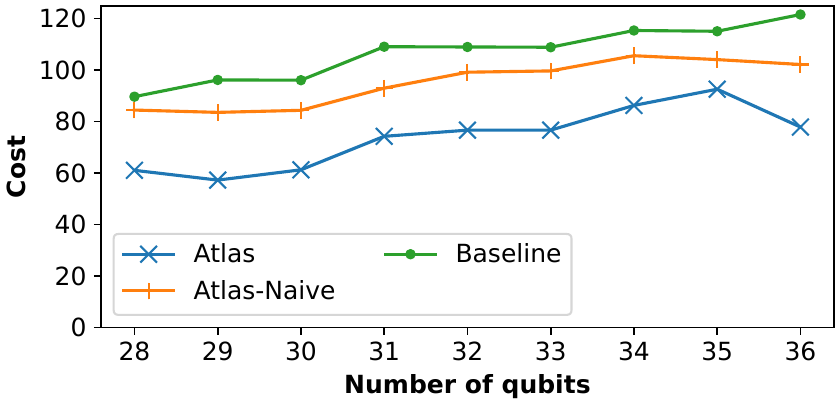}
\caption{The total execution cost of different kernelization algorithms on the circuit \tcd{wstate}.}
\label{fig:dp_wstate}
\end{figure}

\begin{figure}[H]
\centering
\includegraphics[width=0.97\linewidth]{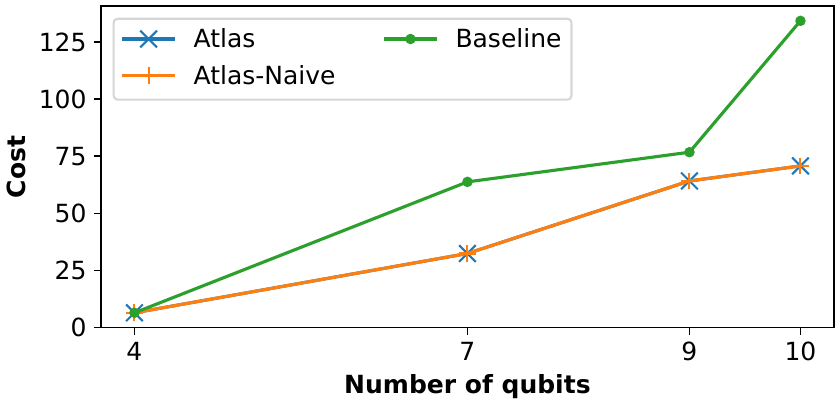}
\caption{The total execution cost of different kernelization algorithms on the circuit \tcd{hhl}.}
\label{fig:dp_hhl}
\end{figure}

\begin{figure}[H]
\centering
\includegraphics[width=0.97\linewidth]{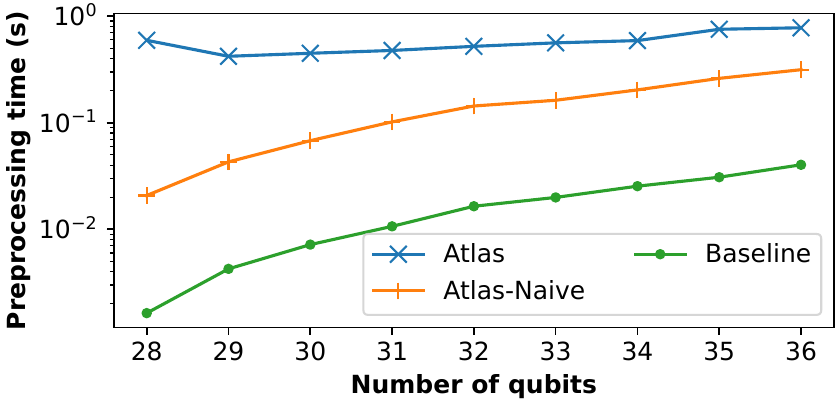}
\caption{The preprocessing time of different kernelization algorithms on the circuit \tcd{ae}.}
\label{fig:dp_time_ae}
\end{figure}

\begin{figure}[H]
\centering
\includegraphics[width=0.97\linewidth]{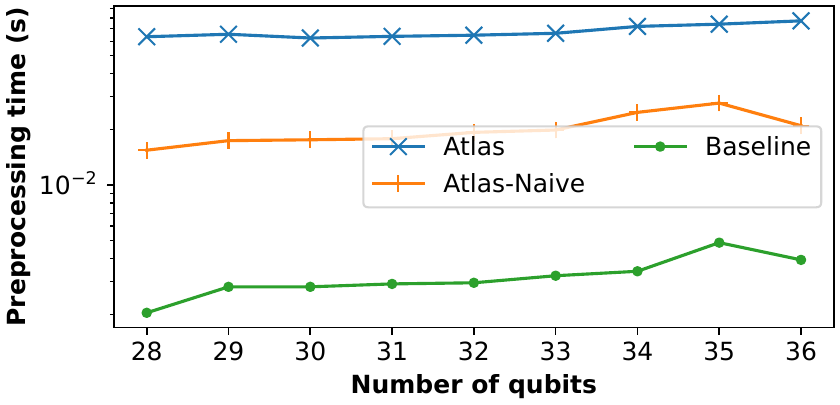}
\caption{The preprocessing time of different kernelization algorithms on the circuit \tcd{dj}.}
\label{fig:dp_time_dj}
\end{figure}

\begin{figure}[H]
\centering
\includegraphics[width=0.97\linewidth]{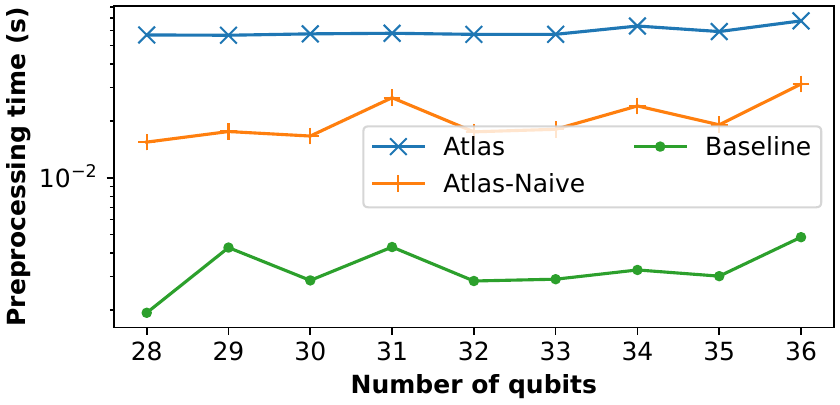}
\caption{The preprocessing time of different kernelization algorithms on the circuit \tcd{ghz}.}
\label{fig:dp_time_ghz}
\end{figure}

\begin{figure}[H]
\centering
\includegraphics[width=0.97\linewidth]{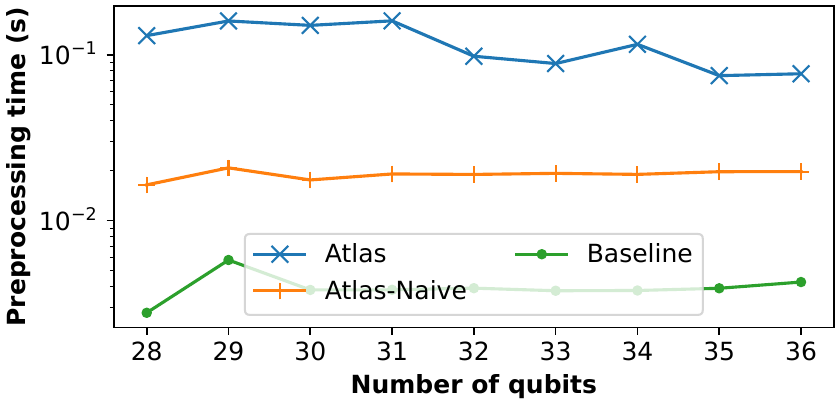}
\caption{The preprocessing time of different kernelization algorithms on the circuit \tcd{graphstate}.}
\label{fig:dp_time_graphstate}
\end{figure}

\begin{figure}[H]
\centering
\includegraphics[width=0.97\linewidth]{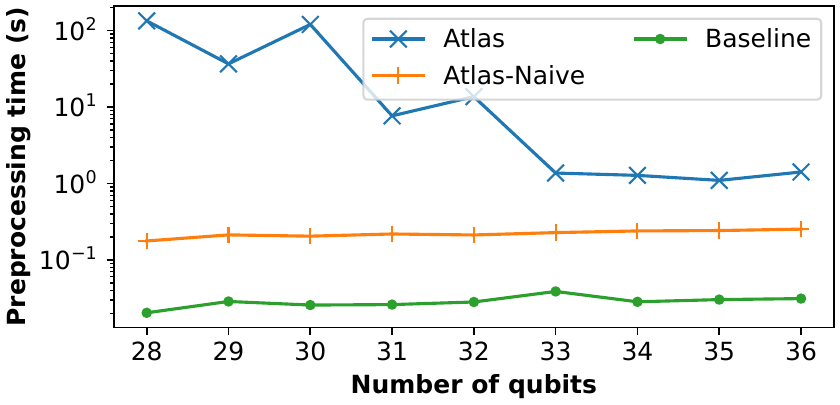}
\caption{The preprocessing time of different kernelization algorithms on the circuit \tcd{ising}.}
\label{fig:dp_time_ising}
\end{figure}

\begin{figure}[H]
\centering
\includegraphics[width=0.97\linewidth]{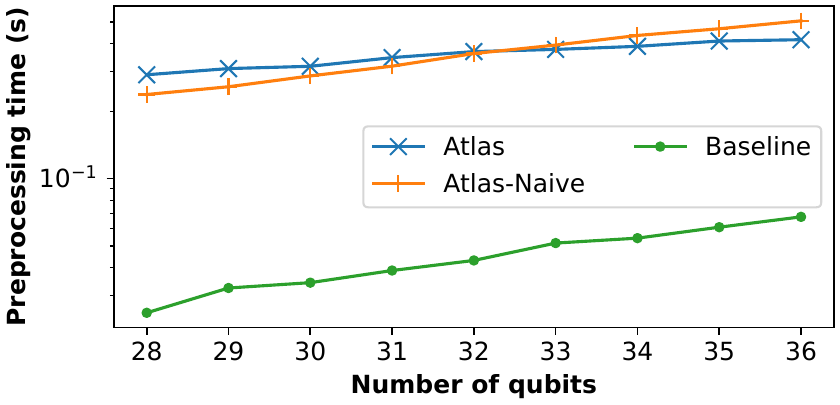}
\caption{The preprocessing time of different kernelization algorithms on the circuit \tcd{qft}.}
\label{fig:dp_time_qft}
\end{figure}

\begin{figure}[H]
\centering
\includegraphics[width=0.97\linewidth]{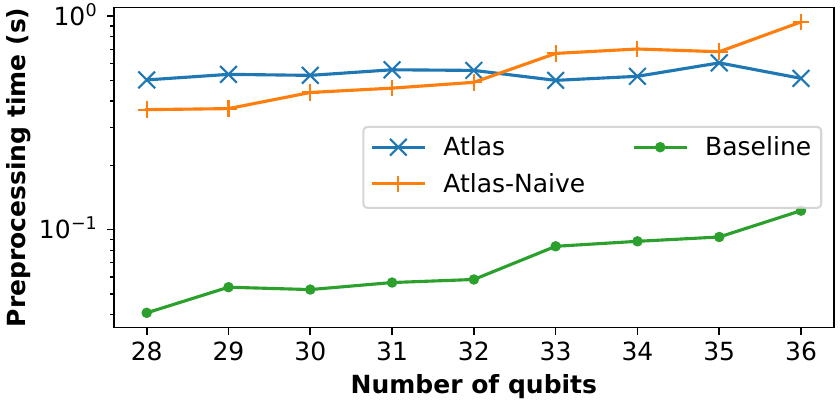}
\caption{The preprocessing time of different kernelization algorithms on the circuit \tcd{qpeexact}.}
\label{fig:dp_time_qpeexact}
\end{figure}

\begin{figure}[H]
\centering
\includegraphics[width=0.97\linewidth]{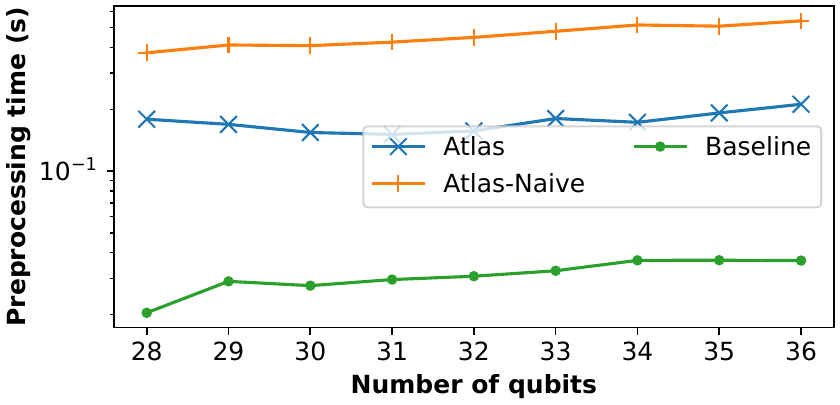}
\caption{The preprocessing time of different kernelization algorithms on the circuit \tcd{qsvm}.}
\label{fig:dp_time_qsvm}
\end{figure}

\begin{figure}[H]
\centering
\includegraphics[width=0.97\linewidth]{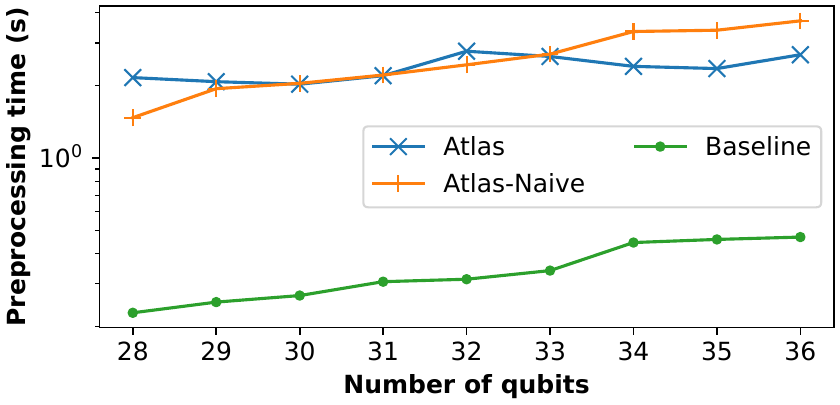}
\caption{The preprocessing time of different kernelization algorithms on the circuit \tcd{su2random}.}
\label{fig:dp_time_su2random}
\end{figure}

\begin{figure}[H]
\centering
\includegraphics[width=0.97\linewidth]{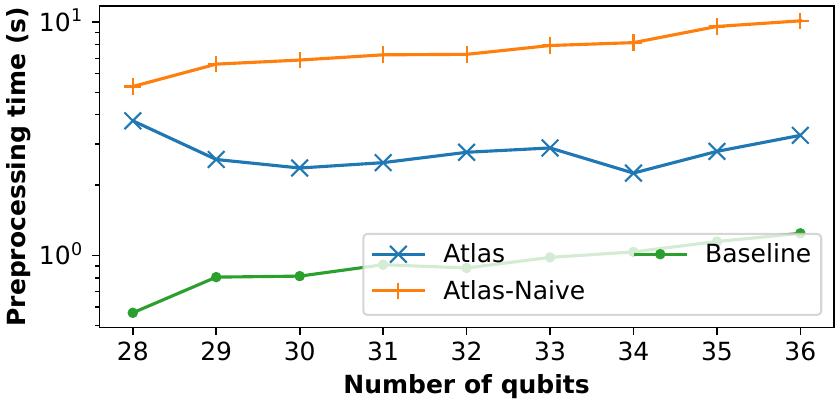}
\caption{The preprocessing time of different kernelization algorithms on the circuit \tcd{vqc}.}
\label{fig:dp_time_vqc}
\end{figure}

\begin{figure}[H]
\centering
\includegraphics[width=0.97\linewidth]{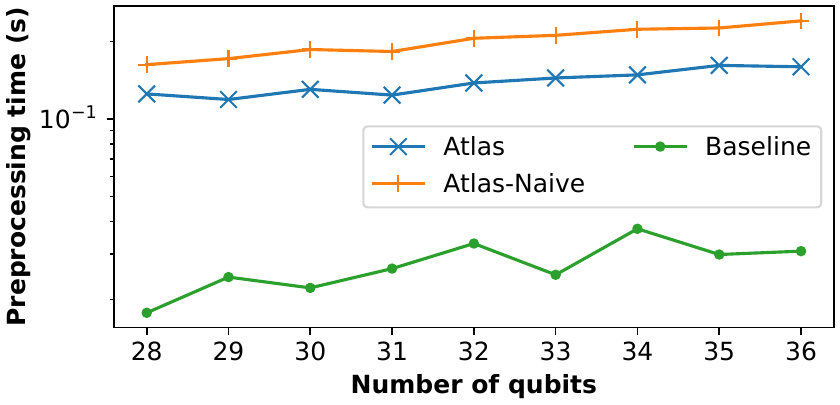}
\caption{The preprocessing time of different kernelization algorithms on the circuit \tcd{wstate}.}
\label{fig:dp_time_wstate}
\end{figure}

\begin{figure}[H]
\centering
\includegraphics[width=0.97\linewidth]{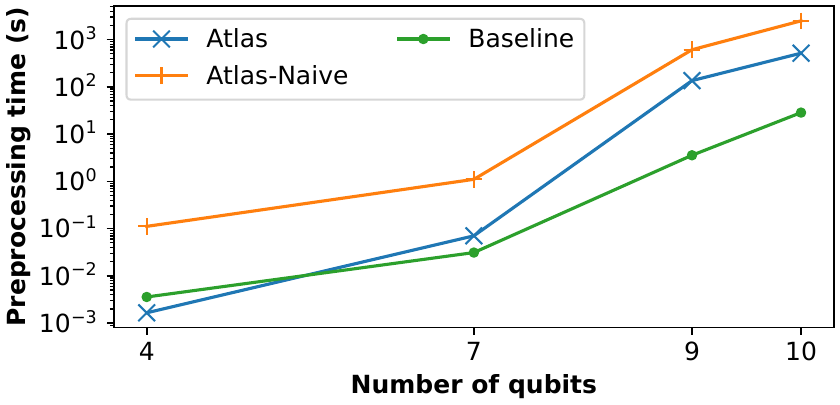}
\caption{The preprocessing time of different kernelization algorithms on the circuit \tcd{hhl}.}
\label{fig:dp_time_hhl}
\end{figure}

\fi

\end{document}